%% file: main.tex
\documentclass[11pt,letterpaper]{article}

\usepackage[margin=1in]{geometry}
\usepackage{cite}
\usepackage{hyperref}
\usepackage{colortbl}
\usepackage{paralist}
\usepackage{pdfpages}
\usepackage{enumitem}
\usepackage{float}
\usepackage{booktabs}
\usepackage[override]{cmtt}
\usepackage{color,xcolor}
\usepackage{graphicx,eso-pic}
\usepackage{boxedminipage}
\usepackage{url}
\usepackage{caption}
\usepackage{subcaption}
\usepackage{xspace}
\usepackage{multirow}
\usepackage{amsmath,amsthm,amstext,amssymb,amsfonts,latexsym}
\usepackage{wrapfig}
\usepackage[linesnumbered,ruled,vlined]{algorithm2e}
\usepackage[noend]{algpseudocode}
\usepackage{algorithmicx}
\usepackage{epstopdf}
\usepackage{mdframed}
\usepackage{cases}
\usepackage[capitalize]{cleveref}
\usepackage[compact]{titlesec}

\input{macro.tex}

\begin{document}

\title{Small Memory Robust Simulation of Client-Server Interactive Protocols over Oblivious Noisy Channels}

\author{T-H. Hubert Chan\thanks{The University of Hong Kong. \texttt{hubert@cs.hku.hk, liangzb@connect.hku.hk}}  \and Zhibin Liang\samethanks
\and Antigoni Polychroniadou\thanks{J.P. Morgan AI Research. \texttt{antigonipoly@gmail.com}}
\and Elaine Shi\thanks{Cornell University. \texttt{runting@gmail.com}}
}

\date{}

\begin{titlepage}

\maketitle

\input{abstract.tex}

\thispagestyle{empty}
\end{titlepage}

\input{intro.tex}

\input{prelim.tex}

\input{overview.tex}

\input{pir.tex}

\input{data_structure.tex}

\input{algo.tex}

\input{proofmulti.tex}

\input{potential_tech.tex}

\input{coding.tex}

\bibliographystyle{plain}
\bibliography{intercode,oram}



\end{document}

%% file: macro.tex
\newcommand{\Point}{{\normalfont \texttt{P}}}

\newcommand{\Array}{{\normalfont \texttt{Arr}}}

\newcommand{\MP}{{\normalfont \texttt{MP}}}
\newcommand{\MPone}{{\normalfont \texttt{MP1}}}
\newcommand{\MPtwo}{{\normalfont \texttt{MP2}}}

\newcommand{\HD}{{\normalfont \texttt{HD}}}
\newcommand{\VP}{{\normalfont \tilde{\texttt{k}}}}
\newcommand{\verif}{{\normalfont \texttt{k}}}
\newcommand{\voteone}{{\normalfont \texttt{v1}}}
\newcommand{\votetwo}{{\normalfont \texttt{v2}}}
\newcommand{\error}{{\normalfont \texttt{E}}}

\newcommand{\poly}{{\normalfont \texttt{poly}}}

\newcommand{\agree}{\ensuremath{\mathfrak{L}}}

\newcommand{\disagree}{{\normalfont \texttt{D}}}
\newcommand{\codeword}{{\normalfont \texttt{W}}}
\newcommand{\counter}{{\normalfont \texttt{CNT}}}

\newcommand{\core}{{\sf core}}

\newcommand{\sset}{{\normalfont S}}

\newcommand{\eps}{\epsilon}

\newcommand{\lhash}{\ensuremath{\mathcal{H}}}
\newcommand{\lseed}{\texttt{S}}
\newcommand{\shash}{\ensuremath{\mathcal{G}}}
\newcommand{\sseed}{\texttt{S}}

\newcounter{note}[section]

\newcommand{\yuetodo}[1]{{\large\color{green}[Yue todo: #1]}}

\definecolor{yue}{rgb}{0.7, 0, 0}

\renewcommand{\yuetodo}[1]{}

\newcommand{\ceil}[1]{\ensuremath{\left \lceil #1 \right \rceil}}
\newcommand{\floor}[1]{\ensuremath{\left \lfloor #1 \right \rfloor}}

\definecolor{darkgreen}{rgb}{0,0.5,0}
\definecolor{lightblue}{RGB}{0,176,240}
\definecolor{darkblue}{RGB}{0,112,192}
\definecolor{lightpurple}{RGB}{124, 66, 168}
\definecolor{grey}{RGB}{139, 137, 137}
\definecolor{maroon}{RGB}{178, 34, 34}
\definecolor{green}{RGB}{34, 139, 34}
\definecolor{types}{RGB}{72, 61, 139}
\definecolor{gold}{rgb}{0.8, 0.33, 0.0}

\definecolor{darkgray}{gray}{0.3}

\definecolor{darkred}{rgb}{0.5, 0, 0}
\definecolor{darkgreen}{rgb}{0, 0.5, 0}
\definecolor{darkblue}{rgb}{0,0,0.5}

\newcommand\markx[2]{}

\newcommand{\Z}{\mathbb{Z}}
\newcommand{\F}{\mathbb{F}}

\newcommand{\rank}{\mathsf{rank}}
\newcommand{\EvalPoly}{\mathsf{EvalPoly}}

\newcommand{\ignore}[1]{}

\newcounter{task}

\newtheorem{theorem}{Theorem}[section]
\newtheorem{remark}[theorem]{Remark}

\newtheorem{corollary}[theorem]{Corollary}
\newtheorem{fact}[theorem]{Fact}
\newtheorem{lemma}[theorem]{Lemma}
\newtheorem{proposition}[theorem]{Proposition}
{
\theoremstyle{definition}
\newtheorem{definition}[theorem]{Definition}
}
\newtheorem{claim}[theorem]{Claim}

\newcounter{cnt:challenge}

\newcommand*\samethanks[1][\value{footnote}]{\footnotemark[#1]}

\SetKwComment{Comment}{//}{}

%% file: abstract.tex
\begin{abstract}

We revisit the problem of low-memory robust simulation of interactive protocols over noisy channels. Haeupler [FOCS 2014] considered robust simulation of two-party interactive protocols over oblivious, as well as adaptive, noisy channels.  Since the simulation does not need to have fixed communication pattern, the achieved communication rates can circumvent the lower bound proved by Kol and Raz [STOC 2013].  However, a drawback of this approach is that each party needs to remember the whole history of the simulated transcript.  In a subsequent manuscript, Haeupler and Resch considered low-memory simulation.  The idea was to view the original protocol as a computational DAG and only the identities of the nodes are saved (as opposed to the whole transcript history) for backtracking to reduce memory usage.

In this paper, we consider low-memory robust simulation of more general client-server interactive protocols, in which a leader communicates with other members/servers, who do not communicate among themselves; this setting can be applied to information-theoretic multi-server Private Information Retrieval (PIR) schemes. We propose an information-theoretic technique that converts any correct PIR protocol that assumes reliable channels, into a protocol which is both correct and private in the presence of a noisy channel while keeping the space complexity to a minimum. Despite the huge attention that PIR protocols have received in the literature, the existing works assume that the parties communicate using noiseless channels. 

Moreover, we observe that the approach of Haeupler and Resch to just save the nodes in the aforementioned DAG without taking the transcript history into account will lead to a correctness issue even for oblivious corruptions. We resolve this issue by saving hashes of prefixes of past transcripts. Departing from the DAG representation also allows us to accommodate scenarios where a party can simulate its part of the protocol without any extra knowledge (such as the DAG representation of the whole protocol). In the the two-party setting, our simulation has the same dependence on the error rate as in the work of Haeupler, and in the client-server setting it also depends on the number of servers.  Furthermore, since our approach does not remember the complete transcript history, our current technique can defend only against oblivious corruptions.

\end{abstract}

%% file: intro.tex
\section{Introduction}
\label{sec:intro}

This paper revisits the problem of low-memory robust simulation
of interactive protocols over a noisy channel that corrupts any $\epsilon$ fraction of the transmitted symbols. Interactive protocols over noiseless communication channels assume that a transmitted message is received as-is. However, transmitted messages are subject to some bounded type of noise in modern communication channels in the presence of environmental or adversarial interference. Given an interactive communication protocol $\Pi$, robust simulation converts $\Pi$ into another communication protocol $\Pi'$ over a noisy channel which is still guaranteed to correctly determine the outcome of the noise-free protocol $\Pi$. The work of Haeupler~\cite{Haeupler14}, considered robust simulation
of a two-party interactive protocol $\Pi$ over a noisy channel
with error rate~$\epsilon$
and could achieve a communication rate of
$1 - O(\sqrt{\epsilon})$ for oblivious corruptions (decided at the onset of the protocol),
and a rate of $1 - O(\sqrt{\epsilon \log \log \frac{1}{\epsilon}})$
for adaptive corruptions (decided during the execution of the protocol based on the communication history).  The approach in~\cite{Haeupler14}
needs to remember the whole history of the simulated transcript,
and a subsequent manuscript by Haeupler and Resch~\cite{abs-1805-06872}
considered low-memory simulation. These works achieve robust communication by adding redundancy i.e., exchanging hash values. When hash values do not match, the parties backtrack. The idea for small storage is to view the protocol $\Pi$ as a computational DAG and only the identities of the nodes are saved for backtracking to reduce memory usage. 
Here are our main contributions.

\begin{compactitem}

\item We consider low-memory robust simulation of more general \emph{client-server}\footnote{
In some literature~\cite{DBLP:journals/talg/ShiCRS17}, the client-server model refers to the case of one server and many clients.
However, in multi-server private information retrieval schemes~\cite{ChorKGS98}, there is one client and more than one server.  Hence,
to avoid confusion, we will use the terminology of one Alice and many Bobs.
} 
interactive protocols,
in which a \emph{leader} (known as Alice) communicates with $m$ other \emph{members} (known as Bobs), who do not communicate with one another;
the special case of $m=1$ lead to a two-party protocol.  

Our communication rate is interesting for small values of~$m$. As we shall see in Section~\ref{sec:pir},
for $m = 2$ or $4$, this setting can be applied
to $m$-server information theoretic private information retrieval schemes~\cite{ChorKGS98,BeimelI01}.

\item  We observe that
even for the special case of $m=1$ and oblivious corruptions, 
low-memory simulation is not previously well understood.
As we shall explain, the approach in~\cite{abs-1805-06872} to just 
save the nodes in the DAG without taking the transcript history into account
 will lead to correctness issue (which is elaborated in Section~\ref{sec:overview}).
We resolve this issue by saving hashes of prefixes of past transcripts.
Our information theoretic technique can defend against only oblivious noisy channels. Note the previous work of \cite{Haeupler14}
that defends against adaptive noisy channels require large memory usage.
As we explain later, in the information-theoretic setting with low memory where the adversary has access to the shared randomness generation it is not clear whether adaptive corruptions can be achieved. 
\end{compactitem}

\begin{theorem}[Our Main Result]
\label{th:main}
Suppose $\Pi$ is an interactive protocol in which the leader Alice
communicates with each of $m$ Bobs through a noiseless channel, and the Bobs do not 
communicate with one another;
suppose further that at most $n$ bits are transmitted in each channel in~$\Pi$.

Then, there is a transformation procedure $\mathfrak{T}$
that, given oracle access to $\Pi$ for some party~$X$ (either Alice or some Bob),
together with parameters
$n$, $m$, $\delta \in 
[\exp(-\Theta(n^{0.249}  \epsilon)), \frac{1}{n}]$
and $0 < \epsilon < 1$ satisfying $m \leq \min\{ n^{0.25}, O(\frac{1}{\epsilon})^{0.199}\}$,
will produce an interface for party~$X$ in another protocol $\Pi'$
such that the following holds.

\begin{enumerate}

\item If protocol $\Pi'$ is run where the communication channel
between Alice and each Bob has \textbf{oblivious} error rate~$\epsilon$, then,
except with probability~$\delta$, $\Pi'$ correctly simulates $\Pi$.

\item In $\Pi'$, for each (noisy) communication channel,
the number of bits transmitted is at most
$(1 + O(m^{2.5} \sqrt{\log m}) \cdot \sqrt{\epsilon}) \cdot n$.

\item If party~$X$ has a memory usage of $M_X$ bits and is involved in~$\kappa$  (= $1$ for each Bob, or $m$ for Alice) communication channels in~$\Pi$,
then its memory usage (in bits) in $\Pi'$ is
at most
$O(\kappa \log \frac{1}{\delta} \log n) + O(\log n) \cdot M_X$.
\end{enumerate}
\end{theorem}

\noindent \textbf{Paper Organization.} While we are trying to give
the precise results of this paper in the above description as soon as possible,
readers unfamiliar with the formal definitions and settings can
first refer to Section~\ref{sec:prelim}.  The research background
and existing results are given in Section~\ref{sec:related}.
An overview of the contribution and methods of this paper is given in Section~\ref{sec:overview};
in particular, we give detailed comparisons with the previous works~\cite{Haeupler14,abs-1805-06872}.
In Section~\ref{sec:pir}, we give an application of Theorem~\ref{th:main}
to multi-server private information retrieval schemes~\cite{ChorKGS98}.

In Section~\ref{sec:data_structure}, we describe the data structure maintained by each party,
and defer proofs related to hashes and randomness in Section~\ref{sec:random}.
The algorithms are described in Section~\ref{sec:algo},
from both Alice's and each Bob's perspectives.
While the overhead and memory usage analysis are relatively straightforward (given in Section~\ref{sec:algo}),
the correctness proof is quite technical.
The high-level proof strategy using potential analysis is described in Section~\ref{sec:proofmultiBob},
while the most technical proofs are deferred to Section~\ref{sec:potential_tech}.

\subsection{Related Work}
\label{sec:related}

The most relevant related works are the aforementioned paper by~Haeupler~\cite{Haeupler14}
and the subsequent manuscript~\cite{abs-1805-06872} that attempted to
perform low-memory robust simulation.  Naturally, the related works described in
them are also related works for this paper, but for the reader's convenience,
we recap some of the works that introduced important concepts and results.

\noindent \emph{Interactive Coding.}
Schulman~\cite{Schulman1996} was the first to design a coding scheme that
tolerates $\epsilon=\frac{1}{240}$ fraction of (oblivious) corruption
with some constant communication rate for interactive channels.
Later, Braverman and Rao~\cite{Braverman2011} improved the tolerated error rate to $\epsilon<1/4$.
Franklin et al.~\cite{Franklin2013} showed
that constant rate cannot be achieved if the error rate is above $\frac{1}{2}$.
As mentioned in~\cite{Haeupler14}, the initial interactive coding schemes are not computationally efficient
because the involved tree codes are complicated to construct.
Using randomization, later results~\cite{Brakerski12,Brakerski13,Gelles2011,GhaffariB2014} achieved polynomial-time coding schemes.

\noindent \emph{Small Error Rate $\epsilon$.}
For binary symmetric channel with random error rate~$\epsilon$,
Kol and Raz~\cite{Kol2013} 
have achieved rate~$1-\Theta(\sqrt{\epsilon\log \frac{1}{\epsilon}})$,
which is optimal for non-adaptive simulation,
i.e., the simulation has a fixed communication pattern.
In contrast, by using adaptive coding schemes,
the aforementioned paper~\cite{Haeupler14}
can circumvent the above lower bound,
and achieve a rate of 
$1-O(\sqrt{\epsilon})$ against oblivious corruptions
and a rate of $1-O(\sqrt{\epsilon \log\log \frac{1}{\epsilon}})$ against adaptive corruptions.
As mentioned in~\cite{Haeupler14},
a subsequent work~\cite{Gelles2015} has achieved 
a better rate of 
$1-O({\epsilon \log \frac{1}{\epsilon}})$
for channels with feedback or erasure channels.

%
%

\noindent \emph{Multi-party Case.}
The multi-party case has also been studied
for random~\cite{DBLP:conf/stoc/RajagopalanS94,DBLP:journals/jacm/BravermanEGH18}
and adversarial~\cite{DBLP:conf/innovations/0002KL15,DBLP:conf/podc/GellesKR19} corruptions.
However, these constructions use large memory and guarantees
only $O(1)$ communication overhead for some restricted error rates.  In contrast,
we consider gracefully degrading communication overhead that tends to 1 as the error rate tends to 0.

\noindent \emph{Small Memory Usage.}
According to~\cite{abs-1805-06872},
their model was influenced by the work of Brody et al.~\cite{Brody2013} on space-bounded complexity.
Moreover, low-memory robust simulation of interactive protocols
has applications in robust circuits design~\cite{Neumann1956,Kleitman1994,Kalai2012},
where the details are described in~\cite{abs-1805-06872}.

%% file: prelim.tex
\section{Preliminaries}
\label{sec:prelim}

We consider robust simulation of an interactive communication protocol $\Pi$
between the following parties:  \emph{Alice} (denoted as $A$) is the \emph{leader}, and there
are $m$  other \emph{members} known as \emph{Bobs} (denoted as $B_i$ for $i \in [m]$).  Each Bob only communicates with Alice
and the Bobs do not communicate with one another.  The special case $m=1$ is the
usual 2-party protocol.

\noindent \textbf{Simplifying Assumptions.} The original protocol $\Pi$ 
proceeds in synchronous \emph{rounds}.
We describe what happens in each round from the perspective of Alice.  At the beginning
of a round, the internal state of Alice will determine, for each $i \in [m]$,
whether she is supposed to send 1 bit to, or receive 1 bit from $B_i$ in this round;
observe that in general, if the protocol requires Alice to send a message of $t$ bits,
we can view this as $t$ rounds, in each of which Alice sends 1 bit.

In the case that Alice is supposed to send 1 bit to $B_i$, her internal state
will determine the value of the bit sent to $B_i$.  After the $m$ bits are transmitted
between Alice and the $m$ Bobs, the internal state of Alice will change accordingly,
and the protocol goes to the next round.   We use $n$ to denote (an upper bound on) the number of rounds of the protocol.  Typically, $m$ is not too large; in particular, our proofs need 
$m \leq \min\{ n^{0.25}, O(\frac{1}{\epsilon})^{0.199}\}$.
Without loss of generality, we can assume that after the protocol $\Pi$ terminates,
Alice continues to send \emph{imaginary} zeroes forever so that we will never run out of
steps to simulate later.

The perspective of each Bob is similar, except that each
Bob only communicates with one Alice.  The number of bits needed to store the state of a party
is known as their \emph{memory usage}.  For each $i \in [m]$, the transcript
between Alice and $B_i$ is a bit string representing the bits transmitted between them.
Clearly, the length-$t$ prefixes (i.e., the first $t$ transmitted bits) of the $m$ transcripts between Alice and 
Bobs determine the internal state of Alice at the end of round $t$.
However, in the original protocol, the memory usage of Alice could be much less than $O(mt)$.

\noindent \textbf{Noisy Communication Channel.}  The original protocol $\Pi$ assumes
that communication is error free.
In a noisy channel, a round might be corrupted, i.e., the value of the transmitted bit is flipped when received
(without the knowledge of the sender); hence, the transcripts of the
two parties can be different.
The error rate of a communication channel
between two parties is $\epsilon$, if, for 
a communication channel through which $n'$ bits are transmitted 
in a protocol,
the number of corrupted bits is at most $\epsilon n'$.
We consider \textbf{computationally unbounded} adversaries in this paper.

\noindent \textbf{Oblivious vs Adaptive Corruptions.}  We assume the adversary knows
(an upper bound on) the number $n$ of rounds of the protocol.
The adversary is \emph{oblivious}, if in which rounds corruption happens is determined
at the beginning of the protocol.  In particular, if each round
is corrupted independently with probability $\epsilon$, the channel is oblivious and
has error rate $O(\epsilon)$ with high probability, by the Chernoff Bound.
On the other hand, an adversary is \emph{adaptive}, if
whether to corrupt a certain round can depend on the bits transmitted in all previous rounds.

\noindent \textbf{Robust Simulation.} 
The goal is to simulate the original protocol $\Pi$
with another protocol $\Pi'$ over noisy communication channels,
without increasing the round complexity or the memory usage too much.
Since communication can be corrupted,
parties might have to roll back their computation when error is detected.
This means that parties might need to increase their memory usage.  

If one only cares about the final computation output,
then a trivial approach might be for all Bobs to send their inputs to Alice who will perform all the computation.
Besides having a large communication overhead in the case of long inputs,
this approach does not respect the intermediate states of a party,
which can be important in some applications, such as private information retrieval that
we will discuss later.

Hence, we define simulation formally as follows.  
We require that throughout the protocol $\Pi'$, each party maintains a pair $(j, s)$,
where $j$ is a round number and $s$ is an internal state in the original protocol $\Pi$.
We assume that each party
has an \emph{imaginary} write-only \emph{state tape} with $n$ entries such that the $\ell$-th entry
is supposed to store their state at the end of the $\ell$-th round of $\Pi$.
At the
end of each round of the new protocol $\Pi'$, the contents of
the $j$-th entry of its state tape is (over)written with $s$.

We use $\delta$ to denote the failure probability of correct simulation.
In this paper, we consider the range $\exp(-\Theta(n^{0.249}  \epsilon)) \leq \delta \leq \frac{1}{n}$;
since $\log \frac{1}{\delta} = \Omega(\log n)$, this can sometimes simplify the expression.

%

\begin{definition}[Robust Simulation]
\label{defn:robust_sim}
A protocol $\Pi'$ correctly simulates $\Pi$, if at the end of $\Pi'$, the contents of all the state tapes
correspond exactly to the internal states of all parties in all rounds in $\Pi$.

In the case that $\Pi$ is randomized,
the joint distribution of the contents of the state tapes 
is the same as that of the internal states of all parties in all rounds in $\Pi$.
\end{definition}


\noindent \textbf{Communication Efficiency.}  Suppose the original protocol $\Pi$ takes
$n$ rounds (which is also the number of transmitted bits)
and it is simulated by $\Pi'$ that takes $n' \geq n$ rounds.
Then, the \emph{communication rate} is $\frac{n}{n'} \leq 1$
and the \emph{communication overhead} is $\frac{n'}{n} \geq 1$.

\noindent \textbf{Some Naive Approaches.}  Observe that we are
aiming to achieve gracefully degrading $1 + \poly(m) \cdot O(\sqrt{\epsilon})$ communication 
overhead, i.e., the overhead tends to 1 as $\epsilon$ tends to 0. 
This rules out certain straightforward approaches that will lead to a
communication overhead of at least some large constant (greater than 1).

\begin{compactitem}
\item \emph{Silent Rounds.}  As mentioned in~\cite{DBLP:conf/podc/GellesKR19},
Hoza observed that the ability to have silent rounds can encode information.  Essentially,
speaking in an odd-numbered round means 1, and speaking in an even-numbered round means 0.
However, this approach will blow up the number of rounds by a factor of 2.

\item  \emph{Cryptography.} If the adversary is computationally bounded, 
then cryptographical tools (such as time-stamped messages together with signature schemes) can be used to detect whether a message has been tampered with.
This can lead to a simpler correction mechanism.  However, the use of cryptography 
will in general lead to a communication overhead of some constant strictly greater than 1.
\end{compactitem}

%% file: overview.tex
\section{Overview of Our Contribution}
\label{sec:overview}

Our approach is based on the work of 
Haeupler~\cite{Haeupler14} on interactive protocols over noisy channels
and a subsequent manuscript~\cite{abs-1805-06872} on low-memory simulation.
In this section, we will give an overview of the approach and highlight
our contributions and how our techniques differ from the previous works.

\noindent \textbf{Intuition for Checksum Bits.}
Suppose, for simplicity, we consider the case where each transmitted bit
is flipped independently with probability $\epsilon$.
The main idea in~\cite{Haeupler14} is that for every \emph{epoch}
of $r$ transmitted bits in $\Pi$, there are $c = \Theta(1)$ checksum bits (which
can depend on previous epochs).  The term ``epoch'' can be applied to both the original protocol $\Pi$
and the transformed protocol $\Pi'$,
where one epoch in $\Pi'$ can simulate at most one epoch in $\Pi$ and have some extra computation and checksum bits.

Suppose further that if an error occurs in some epoch in $\Pi'$,
the error will eventually be detected by the checksum bits (maybe in subsequent epochs).
Even with these lenient assumptions, Haeupler~\cite{Haeupler14} has derived
a lower bound on the communication overhead of this approach as follows.  The checksum bits alone
will lead to a communication overhead of at least $\frac{r + c}{r}$.
Because the channel has error rate $\epsilon$,
each epoch has no error with probability around $1 - \epsilon r$,
which means each epoch has to be transmitted, in expectation, $1 + O(\epsilon r)$ times.
Hence, the communication overhead
is at least $\max\{1 + O(\epsilon r), 1 + \frac{c}{r}\} = 1 + O(\sqrt{c \epsilon})$,
which is achieved when $r = \Theta(\sqrt{\frac{c}{\epsilon}})$.
We shall later see that when there are $m$ Bobs, both $r$ and $c$ will depend on~$m$.

\noindent \textbf{Meeting Point Based Backtracking.} As credited by Haeupler~\cite{Haeupler14},
their protocol uses \emph{meeting points} for backtracking, which
originated in Schulman's first interactive coding paper~\cite{Schulman92}.
The idea is that when a party is currently at the end of epoch~$p$ in $\Pi$
of the simulation, then for each $k$ that is a power of two,
the party can possibly revert to the end of the epoch indexed by the following meeting point:
$\mathsf{MP}_k(p) := k \cdot \floor{\frac{p}{k}} - k$.  The following ideas are crucial
to the success of this method.

\begin{compactitem}
\item \emph{Shared Randomness for Hashing.}  To reduce communication overhead,
two parties can find out if they have a common meeting point by comparing
hashes of the relevant prefixes of their transcript histories.  Although inner-product hash
uses randomness whose length is the same as the object to be hashed,
one does not need totally independent randomness.  In fact, a seed that is logarithmic in length
can be used to generate biased randomness~\cite{Naor1993} which will be good enough.
Hence, at the beginning of the protocol, parties just need to robustly agree on this shorter seed.

\item \emph{Remembering the Transcript.}  It is important that each party remembers his complete transcript, i.e., history
of communication.
Even though two parties might have different transcript prefixes of a certain length,
they might think that the corresponding meeting point is valid, because of
hash collision or corruption.  However, since the original transcript is still available to each party,
as the protocol proceeds, a different randomness will eventually reveal the discrepancy.

\end{compactitem}

Although the idea of meeting points is intuitive, an intricate potential function
was used to argue that the simulation can be completed with small communication overhead.
Intuitively, a potential function is used to keep track of the progress of the simulation and corruptions made so far
such that if the potential function is large enough, then the original protocol $\Pi$ is completed.
As we shall see, for small memory, we will use an even more sophisticated potential function.

\noindent \textbf{Challenges for Low-Memory Simulation.}  The manuscript~\cite{abs-1805-06872} has attempted to lower the memory usage of the simulation.
The intuition is that if the original protocol $\Pi$ can be represented by a DAG
with $s$ nodes, then $O(\log s)$ bits is needed to store each meeting point.
Since at most $O(\log n)$ meeting points are stored,
the protocol can be simulated with $O(\log s \log n)$ bits of memory.
However, we think there are several issues with this approach.

\begin{compactitem}

\item \emph{Too restrictive model.}  It is assumed that the original protocol $\Pi$ is somehow transformed
into a DAG that is known by all parties, who must all have the same memory usage.
However, in other settings such
as the client-server model, the client party might only have partial knowledge of the original protocol $\Pi$
that is relevant to him.  Hence, a more desirable transformation should allow a party
to simulate its part of the original protocol $\Pi$ without any extra knowledge (such as the DAG representing the whole protocol).
Furthermore, in the original protocol $\Pi$, the client party might have a much smaller memory usage
than the server party, and it would be undesirable if every party needs to have the same memory usage in the simulation.

\item \emph{Just remembering the state is not enough.}  A careful analysis of the manuscript~\cite{abs-1805-06872} 
reveals that in their approach,
only a node is saved in a meeting point, but the information on how the node was reached (i.e., the transcript)
is not saved or used to compute the hash of that meeting point.  After all,
the whole point of a low-memory simulation is to avoid remembering the transcript.

However, this poses a serious correctness issue.
The following simple example suggests that the hash for a meeting point should involve the corresponding prefix of the transcript.  Suppose each of Alice and Bob has an input bit,
which is transmitted to the other party in two rounds,
after which the protocol is in node~0 if the bits agree, and in node~1, otherwise.
Suppose the parties initially have different bits.  Then, if there is no corruption,
the protocol should be in node 1 after two rounds.
However, if both transmissions are corrupted,
then each party thinks (incorrectly) that the other party has the same bit.
If each epoch has 2 rounds, then both parties will think that they are in node~0 after the first epoch and computing the hash just based on the node identity will not detect the mistake.

\end{compactitem}

\noindent \textbf{Our Solution: Saving Hashes of Transcript Prefixes.}  In addition to using hashes to
produce checksums, we observe that they can be used to reduce memory storage for transcript history.
To detect inconsistent transcript history, two parties do not actually need to know their exact past transcripts,
but just need to tell that they are different.  Therefore, instead of saving their whole transcript history,
when a party saves a meeting point, it just needs to save the hash of the corresponding transcript prefix.
Even though the underlying intuition is simple, we still need to pay attention to the following details.

\begin{compactitem}
\item \emph{Sharing Randomness with Low Memory.}  Recall that to produce one bit of inner product hash requires
a randomness whose length is the same as that of the object to be hashed.  Moreover, since we only remember
the hash of a transcript prefix, if hash collision happens for inconsistent transcript history between two parties,
then there is no way to recover correctness.  Hence, to achieve failure probability of at most $\delta$,
the hash for a transcript prefix needs to have at least $\Omega(\log \frac{1}{\delta})$ bits.

Even though biased randomness can be produced with a seed of logarithmic length~\cite{Naor1993},
we cannot afford too much space to store the stretched randomness explicitly.
In Section~\ref{sec:random}, we will describe a low-memory variant of generating biased randomness
that unpacks random bits from the seed as we need them.

\item \emph{Saving Transcript Hashes Can Defend Against Only Oblivious Corruptions.}
At first sight, since the adversary can observe the shared randomness,
it is natural that the hash collision analysis is valid only for oblivious corruptions.
However, the analysis in~\cite{Haeupler14} treated an adaptive adversary
as a collection of oblivious adversaries.  By slightly increasing the length
of the checksums, a union bound over the collection of oblivious adversaries
can still make the hash collision analysis work.

However, such an approach cannot work for the hashes of transcripts.  The reason is
that in~\cite{Haeupler14}, each party remembers its complete transcript history.  Hence,
even when a hash collision occurs for some checksum, the underlying
discrepancy can still be potentially discovered later.  On the other hand,
to carry out such a union bound for the hashes of transcripts,
the length of the hash would have to be as long as the transcript itself, which
defeats the purpose of achieving low memory in the first place.  This is a major reason
why our current approach only works for oblivious corruptions.

It was suggested in~\cite{Haeupler14} that encryption can
be used against an adaptive adversary that is computationally bounded.
Specifically, one could first encrypt the seed for the shared randomness.  However, even if the adversary does not
know the secret biased randomness initially, once hashes are being produced,
the adversary can learn some information about the biased randomness to make future corruptions.
Hence, perhaps as future work, more careful analysis is required to claim that
using cryptography can defend against adaptive adversaries that are computationally bounded.

\end{compactitem}

\noindent \textbf{Generalization to Client-Server Setting.}
After replacing the hashes for nodes in~\cite{abs-1805-06872}
with hashes for transcript prefixes,
the meeting point based backtracking approach will work for the simulation
of two-party protocols.  When we adapt this approach to a client-server interactive protocol,
the perspective of each Bob (Algorithm~\ref{alg:ComputeOblivious}) is essentially the same as if he is in a two-party protocol, because he can only see one Alice.
On the other hand, for Alice to proceed the simulation,
she needs to make sure that she and all other Bobs have a consistent transcript history;
moreover, if she needs to roll back the computation,
she also needs to make sure that there is a common meeting point.

After adapting the flow structure of the simulation for Alice,
the analysis of communication overhead and memory usage follows 
directly from the algorithm parameters.  The difficulty is how to choose
the parameters to ensure that the simulation is correct with the desired
probability.  When we adapt the potential function analysis in~\cite{Haeupler14}
to multiple number of Bobs, the constants will have a dependence on $m$,
which will eventually affect the communication overhead.  From Alice's perspective,
since there are $m$ channels, each round is $m$ times more likely to be corrupted;
hence, the overhead should be at least $1 + \Omega(\sqrt{m \epsilon})$,
while our current approach achieves $1 + \widetilde{O}(m^{2.5} \cdot \sqrt{\epsilon})$.

As we shall see, this dependence on~$m$ comes from the complicated potential analysis
in Section~\ref{sec:potential_lb} for low-memory usage, which is adapted
from~\cite{Haeupler14,abs-1805-06872}.  In particular, we make some of the arguments
to accommodate for unavailable meeting points more explicit.  As future work, the dependence
on $m$ can probably be improved by a better potential analysis.

%% file: pir.tex
\subsection{Application to Multi-Server Private Information Retrieval Schemes}
\label{sec:pir}

Private information retrieval (PIR)~\cite{ChorKGS98} allows a client to
outsource storage of some \textbf{read-only} data on
non-colluding servers such that the client
can access part of the data without each server knowing which part
of the data the client really needs.

In the seminal work~\cite{ChorKGS98},
some $N$-bit array $\Array[1..N]$ is stored in
each of $m$ servers.  To access some bit in $\Array$ 
indexed by~$i \in [N]$,
the client sends a (possibly randomized) message
to each of the $m$ servers,
and each of the servers responds with a message.
From the responding messages (and possibly together with the original
messages sent to the servers),
the client can decode the bit $\Array[i]$,
but each of the servers cannot learn the index~$i$, even with
unbounded computational power.
Since this paper concerns the case where the overhead is close to 1,
we consider PIR schemes with explicit constants in the guarantees.
The following results are for $m=2,4$ servers,
and the communication complexity refers to the total number
of bits transmitted to request one bit of $\Array$.

\begin{fact}[Multi-Server PIR Schemes~\cite{ChorKGS98,BeimelI01}]
\label{fact:pir}
There exist information-theoretic $m$-server PIR schemes to store an $N$-bit array such
that the communication complexity to request each bit
is:
\begin{compactitem}

\item $28 \sqrt[4]{N} + 4$ for $m=4$~\cite{ChorKGS98}, i.e.,
$7 \sqrt[4]{N} + 1$ bits between the client and each server.

\item $4 (6N)^{\frac{1}{3}} + 2$ for $m = 2$~\cite{BeimelI01}, i.e.,
$2 (6N)^{\frac{1}{3}} + 1$ bits between the client and each server.

\end{compactitem}
\end{fact}

We use our technique to consider the case when the communication channel
between the client and each PIR server can be corrupted by oblivious noise.
Since each request induces only a constant number of messages (where each
message has $O(\sqrt[4]{n})$ bits), it follows that
the memory usage is of the same order as the transcript history.  However,
our low-memory technique will be useful when we consider a sequence of
PIR requests together.  One important property is that the PIR scheme
is history independent, i.e., neither the client nor the servers need to retain 
any information about previous requests.  This is crucial when the communication channel
is noisy, because corrupted requests in the past will not compromise the security of
future requests.  Our Theorem~\ref{th:main} gives the following corollary.

\begin{corollary}[Multi-Server PIR on Noisy Channels]
Suppose a client runs a program that needs
to make $T$ PIR requests to an $N$-bit outsourced array,
where each communication channel has oblivious error rate $\epsilon$
and the desired overall failure probability is $\delta$.  Then,
applying Theorem~\ref{th:main}, if
the parameters satisfy the hypothesis with (1) $n = T \cdot (7 \sqrt[4]{N} + 1)$
and $m=4$ servers, or (2)
$n = T \cdot (2 (6N)^{\frac{1}{3}} + 1)$ and $m = 2$ servers,
the program will run correctly except with probability~$\delta$, PIR security
is still maintained,
and the following are also achieved.

\begin{enumerate}

\item The number of bits transmitted between the client
and each server is $(1 + O(\sqrt{\epsilon})) n$.

\item If the original memory usage of the program is $M_c$ bits with noiseless PIR channels,
then the new memory usage is $O(\log \frac{1}{\delta} \log n + M_c \cdot \log n)$ bits.

\item In addition to storing the $N$-bit array,
the memory usage of each server is $O(\log \frac{1}{\delta} \log n + \frac{n}{T} \log n)$.
\end{enumerate}

\end{corollary}

%% file: data_structure.tex
\section{Data Structure for Hashing Transcripts}
\label{sec:data_structure}

In this section, we show how each party can avoid storing the complete simulated transcript of the
original protocol~$\Pi$.
Instead, for each meeting point, a hash for the corresponding transcript will be sufficient.

\noindent \textbf{Notation.}  Recall that each epoch contains $r$ rounds of the original protocol $\Pi$,
which runs in at most $n$ rounds.  We use the array $\sigma[1..p]$ to denote a party's view of the transcript in $\Pi$
up to epoch~$p$ over
a (pair-wise) communication channel, i.e., $\sigma[i]$ (or $\sigma_i$) consists of the $r$ bits sent or received by that party in epoch~$i$
of the original protocol.  As discussed in Section~\ref{sec:overview},
a party needs to remember at least part of $\sigma$ to ensure correctness of simulation.
Moreover, as we shall see, the simulation is performed for $R = \ceil{\frac{n}{r}} + \Theta(m^5 n \epsilon)$ epochs, which 
is slightly larger than $\frac{n}{r}$ to ensure correct termination.

Each Bob is involved in only one communication channel. However, Alice is involved with $m$ Bobs, and needs to 
keep track of the $m$ communication channels.  

\noindent \textbf{Meeting Point.}  A meeting point is an index~$p$ that
represents the simulation at the end of epoch~$p$ of the original protocol.
The data structure for saving a meeting point consists of the following:

\begin{compactitem}

\item The index $p$ itself (which takes $\log_2 R$ bits).

\item For each involved communication channel,
an $o$-bit \emph{long} hash $H_p$ of its transcript $\sigma[1..p]$.
We shall discuss the long hash in more details in Section~\ref{sec:long_hash}.

\item The state of the party at the end of epoch~$p$ of the original protocol~$\Pi$.

\end{compactitem}

\begin{claim}
For Alice, each meeting point takes $\log_2 R + M_A + m \cdot o $ bits,
where $M_A$ is the memory usage of Alice in the original $\Pi$.
\end{claim}

\subsection{Long Hash for Storing Transcript}
\label{sec:long_hash}

\noindent \textbf{Pre-shared Randomness.} Since the transcript can have $R$ epochs,
its length can be up to $R \cdot r$ bits.  Recall that we wish to produce an $o$-bit long hash,
where each hash bit is obtained using inner-product hash.  Therefore,
for each communication channel, we will need some pre-shared randomness $\mathcal{S}$
consisting of $o \cdot R \cdot r = \Theta(n \cdot o)$ bits between the corresponding two parties.

Observe that the communication overhead will be too large if $\mathcal{S}$ is transmitted directly.
As observed in~\cite{Haeupler14}, we do not need total independence for $\mathcal{S}$
and it is sufficient for $\mathcal{S}$ to be $\rho$-biased, where $\rho = 2^{-o}$.
The reader can refer to~\cite{Naor1993} for background on biased randomness.
As far as understanding this paper, one just needs to know that the two involving parties can agree on some shorter random seed $\core^*$, from which $\mathcal{S}$ can be extracted
such that the resulting hash collision probability is comparable to truly independent randomness.
The following result states the length of $\core^*$.

\begin{fact}
[Seed for Biased Randomness~\cite{Naor1993}]
\label{fact:generatebiasedrand}
Generating $\rho$-biased randomness of $q$ bits can be done using $\Theta(\log q + \log\frac{1}{\rho})$ independent random bits.
\end{fact}

However, a party cannot afford to store $\mathcal{S}$ explicitly.
Given $\core^*$, one should be able
to extract each bit from $\mathcal{S}$ when needed.
Hence, we will use a weaker version of Fact~\ref{fact:generatebiasedrand},
where the seed $\core^*$ has $\Theta(\log\frac{1}{\rho} \log q)$ bits.
In Section~\ref{sec:random}, we shall describe how the following
subroutines are achieved.  Most of them are standard in the literature,
but we need to pay attention to low memory usage.

\begin{compactitem}

\item \textsc{RobustSend}($\core^*, t$) and \textsc{RobustReceive}($\ell, t$).
Provided that the channel corrupts at most $t$ bits,
this protocol allows the sending party to robustly send an $\ell$-bit string
$\core^*$ to the receiving party by transmitting $\Theta(\ell + t)$ bits over the channel,
where the memory usage of both parties is $O(\ell + \log \frac{t}{l})$ bits.
%

\item $\mathcal{S} \gets$ \textsc{RandInit}($\core^*, q \in \Z^+, \rho \in (0,1)$).  With a seed
$\core^*$ of $\Theta(\log \frac{1}{\rho} \log q)$ bits,
a $\rho$-biased random string $\mathcal{S}$ of $q$ bits is implicitly initialized,
using only $\Theta(\log \frac{1}{\rho} \log q)$ bits of memory storage.

\item $\mathcal{S}.$\textsc{ExtractBit}($i \in [1..q]$). This subroutine returns the $i$-th bit of $\mathcal{S}$,
using $O(\log \frac{1}{\rho} \log q)$ bits of extra memory.

\item $\mathcal{S}.$\textsc{ExtractBlock}($p \in [1..\frac{q}{b}], b$).
If one views $\mathcal{S}$ as $\frac{q}{b}$ blocks (where each block has $b$ bits),
this subroutine  returns the $b$ bits
in the $p$-th block.
This can be achieved by calling
\textsc{ExtractBit} with indices in $[(p-1)b + 1, pb]$.
 Observe that we may omit the argument $b$, if the block size $b$ is clear from context.  

\end{compactitem}

These subroutines will be used as follows.
For hashing the transcript,
we use a random seed $\core^*$ with $\Theta(o \log (R \cdot  r \cdot o))$ bits long,
which will be robustly transmitted between two parties with at most $2 n \epsilon$ corruptions\footnote{
As we shall see in Lemma~\ref{lemma:overhea},
the overhead is at most 2 and this means there can be at most $2n\epsilon$ corruptions
throughout the whole simulation.}.
Each party will 
implicitly stretch $\core^*$ to a $2^{-o}$-biased string 
$\mathcal{S} \gets$ \textsc{RandInit}($\core^*, R r o, 2^{-o}$)
of $R \cdot  r \cdot o$ bits.
One can view $\mathcal{S} = (S^{(p)}: p \in [R])$ as $R$ blocks,
where each block has $r \cdot o$ bits.  
Given an index $p$, the block $S^{(p)} \gets$ $\mathcal{S}.$\textsc{ExtractBlock}($p$)
can be retrieved.  Therefore, the hash $H_p$ for $\sigma[1..p]$ can be recursively computed as follows.

\begin{definition}
[Inner-Product Hash]
\label{def:transcripthash}
Given the previous hash $H_{p-1}$ for $\sigma[1..p-1]$ and 
the block
 $S^{(p)}=(\lseed_1,..,\lseed_o)$ of $o \cdot r$ bits,
we can compute an $o$-bit hash
$H_p=\lhash(H_{p-1},\sigma_p,S^{(p)})$ as follows for each output bit $i\in[1..o]$:
$$H_p[i]=H_{p-1}[i]+\left\langle \sigma_p,\lseed_{i} \right\rangle.$$
\end{definition}

\begin{remark}
Because we use inner-product hash,
the hash of $\sigma[1..p]$ will be the same if we append trailing zeroes to
the end of the transcript.  Hence, the following proposition~\cite[Corollary 6.2]{Haeupler14} considers hash collision
between different transcripts of the same length.
\end{remark}

\begin{proposition}[Hash Collision]
\label{prop:hash_collision}
Suppose $\sigma$ and $\widehat{\sigma}$ are bit strings of the same length,
but differ in at least one bit.  Moreover, suppose $\rho$-biased
randomness $\mathcal{S}$ is used to produce $o$-bit inner-product hashes $H$ and $\widehat{H}$ for $\sigma$ and $\widehat{\sigma}$, respectively,
as in Definition~\ref{def:transcripthash}.
Then, with probability at least $2^{-o} + \rho$,
the hashes $H$ and $\widehat{H}$ are different.
\end{proposition}

\begin{lemma}[Hashing Transcripts]
\label{lemma:long_hash}
Consider an oblivious noisy communication channel between two parties, who
produce $o$-bit long hashes of their simulated transcripts (in the original $\Pi$) using pre-shared $\rho$-biased
randomness $\mathcal{S}$ as in Definition~\ref{def:transcripthash}, where $\rho = 2^{-o}$.
Then, over $R$ epochs of the simulation, 
the probability that the two parties ever store meeting points $(p, H_p)$ with the same $p$ and hash $H_p$ but 
resulting from different transcripts (of $p$ epochs in $\Pi$)
is at most $\frac{2 R^2}{2^o}$.
\end{lemma}

\begin{proof}
We first condition on any randomness apart from $\mathcal{S}$ that might be used in the $R$ epochs or the oblivious corruption.
Since each party will try to store at most $R$ meeting points,
there are at most $R^2$ pairs of meeting points from the two parties.
Using $\rho = 2^{-o}$ in Proposition~\ref{prop:hash_collision},
each such pair will lead to a hash collision with probability at most $\frac{2}{2^o}$.
Hence, the union bound over all $R^2$ pairs gives the result.
\end{proof}

In view of Lemma~\ref{lemma:long_hash},
except with the stated failure probability of long hash collision,
during the simulation, each party is able to save enough information
to potentially tell whether its meeting point and that of another party
correspond to the same transcript prefix.  Observe
that the failure event in Lemma~\ref{lemma:long_hash} can compromise the correctness
of the simulation.  By considering the union bound over all $m \leq n$ channels
in Lemma~\ref{lemma:long_hash}, the following corollary determines
the length of the long hash.

\begin{corollary}[Length of Long Hash]
\label{cor:long_hash_len}
For $0 < \delta \leq \frac{1}{n}$,
by choosing the long hash length $o = \Theta(\log \frac{1}{\delta} + \log n) = \Theta(\log \frac{1}{\delta})$,
the failure probability of long hash collision is at most $\delta$.
\end{corollary}

\subsection{Short Hash for Comparing Meeting Points (and Other Objects)}

Observe that the $o$-bit long hash in Section~\ref{sec:long_hash}
is too long to be sent directly to the other party.  Whenever two parties
wish to compare objects that are too large to be transmitted directly,
they will compute a $c$-bit \emph{short} hash for an object first and
transmit the short hash instead.  Here, $c = \Theta(\log m)$ is
not too large, and hence, it is quite possible to have a hash collision.

\noindent \textbf{Parameters for Short Hash.} We will also use inner-product hash. 
As we shall see, the largest object we need to compare is a pair $(p, H_p)$ associated
with a meeting point, which takes $L := \log_2 R + o$ bits.  Therefore,
to produce a $c$-bit hash, we implicitly need some randomness $S$ with 
$cL$ random bits,
where $S$ can be $2^{-c}$-biased.

\noindent \textbf{Limited Dependence.}  As we shall see later, we need to apply
Chernoff Bound over different epochs of the transformed protocol $\Pi'$.  Since
we are not aware of a simple way to apply Chernoff Bound with biased randomness,
the randomness for producing the short hash will need to be refreshed frequently.
However, refreshing the randomness for every epoch will lead to too large communication overhead.
As a compromise, the randomness $S$ for producing the short hash is regenerated
independently every $I$ epochs in the transformed protocol $\Pi'$, which we call a \emph{phase};
we shall see later that $I = \Theta(\log o)$.

\noindent \textbf{Shared Randomness in Each Phase.}  At the beginning
of each phase, the two parties share biased randomness $S \gets$ \textsc{RelaxedShareRand}
as follows.

\begin{definition}[\textsc{RelaxedShareRand}]
\label{defn:relaxedrand}
The subroutine \textsc{RelaxedShareRand}
consists of the following steps:
\begin{enumerate}
\item Alice picks a common $\core$ with $O(c \log c L)$ bits
and use \textsc{RobustSend}$(\core, I)$ to transmit the message
to  every Bob over the channel.  Observe that each Bob can decode
the correct $\core$, if the channel corrupts at most $I$ bits.

\item Using its perceived seed $\core$, each party
uses \textsc{RandInit}$(\core, cL, 2^{-c})$ to initialize the $2^{-c}$-biased randomness.
Then, it can use \textsc{ExtractBit} repeatedly to extract
the $cL$-bit $S$.
\end{enumerate}
\end{definition}

\ignore{
\begin{lemma}[Number of Epoch with Corrupted Randomness]
There are at most $\Theta(n \epsilon)$ epochs in which the randomness $S$
for short hash is corrupted.
\end{lemma}

\begin{proof}
To corrupt an $S$ in a phase takes at least $I$ corrupted bits.  Hence,
at most $\frac{2 n \epsilon}{I}$ phases can have their $S$ corrupted.
Since each phase has $I$ epochs, the result follows.
\end{proof}
}

\begin{definition}
[Short Hash Function]
\label{def:innerproducthash}
We define the short hash function $\shash$ as follows.
Given randomness $S=(\sseed_1,..,\sseed_c)$ of $c L$ bits
and an object $O$ of size at most $L$ bits,
$G=\shash(O,\sseed) \in \{0,1\}^c$ is as follows: for each output bit $i\in[1..c]$,
$G[i]=\left\langle O,  \sseed_{i}[1..|O|] \right\rangle$.
\end{definition}

Similar to Proposition~\ref{prop:hash_collision},
the probability of short hash collision using biased randomness
is given as follows.

\begin{proposition}[Short Hash Collision]
Suppose the channel makes at most $I$ corruptions
when $S \gets$ \textsc{RelaxedShareRand} is shared between the two parties.
Then, if two parties use its own $S$ to compute short hashes
of different objects with the same size (at most $L$) as in Definition~\ref{def:innerproducthash},
 the probability of hash collision is at most $\frac{2}{2^c}$.
\end{proposition}

%% file: algo.tex
\section{Description of Robust Simulation over Noisy Channel}
\label{sec:algo}

We use the same simulation approach as in~\cite{abs-1805-06872}.
As mentioned before, the major difference is that for a meeting point,
our approach saves the hash of the corresponding transcript prefix (in addition
to the corresponding internal state in $\Pi$),
while just saving the node in the computation DAG as in~\cite{abs-1805-06872}
will lead to correctness issues.  The reader can first refer to each
Bob's version in Algorithm~\ref{alg:ComputeOblivious},
which only deals with one communication channel and essentially has the
same structure as the algorithm in~\cite{abs-1805-06872};
the auxiliary variables $\alpha$ and $\beta$ are
used for correctness analysis in Section~\ref{sec:proofmultiBob}.
Alice's version is given in Algorithm~\ref{alg:AliComputeOblivious},
which has a more complicated control flow structure, because simulation can continue
only if the transcripts are consistent for all $m$ communication channels.

\noindent \textbf{Hash Dictionary.} A party will, for each involved communication channel,
maintain a \emph{hash dictionary} that is denoted by $\HD$.  If an index $p \in [1..R]$ 
is saved in $\HD$, then the following information is stored:
\begin{compactitem}
\item The index $p$ itself.

\item For each Bob, the $o$-bit long hash of the transcript $\sigma[1..p]$ as in Definition~\ref{def:transcripthash},
which we denote as $\HD[p]$ as a short hand;
if the index~$p$ is not saved, we use the convention that 
$\HD[p] = \bot$.

For Alice, she needs to store $m$ such $o$-bit long hashes, because she communicates with $m$ Bobs.
For each $i \in [m]$, we use $\HD_i[p]$ to denote the long hash of the transcript of $\sigma_i[1..p]$
with $B_i$.  Again, if $p$ is not saved, we use the convention that $\HD_i[p] = \bot$.

\item The internal state of the party in the original protocol $\Pi$ at the end of epoch $p$.
\end{compactitem}

\noindent \textbf{Available Meeting Points.}  Suppose a party has
finished the simulation of $\Pi$ up to epoch~$\Point$.  Then, 
in addition to~$\Point$, an index~$p \neq \Point$ is saved in $\HD$
only if there exists an integer $j \geq 0$ such that $p = 2^j \floor{\frac{\Point}{2^j}} - 2^j$.
This invariant is maintained by Line~\ref{line:removepoints} in Algorithm~\ref{alg:ComputeOblivious}.
Observe that the converse is not true, i.e.,
it is possible that the index for some $j$ in the above form is not saved.

\begin{remark}
At any moment, the number of saved meeting points is at most $\log_2 R = O(\log n)$.
\end{remark}

\noindent \textbf{Scale of Meeting Points.}  Each party maintains some counter~$\verif$
that keeps track of how far the simulation needs to be rolled back.  Whenever $\verif$ reaches
a power of two, then there are two candidate meeting points of scale~$\verif$ with respect to the current epoch~$\Point$:

$\MPone_{\Point}(\verif) = \verif \floor{\frac{\Point}{\verif}}$
and $\MPtwo_{\Point}(\verif) = \verif \floor{\frac{\Point}{\verif}} - \verif$.

By considering the binary representation of $\Point$,
$\MPone_{\Point}(\verif)$ corresponds to setting all its $\log_2 \verif - 1$ least significant digits to zero.
One can also verify that
for any $\verif \geq 2$ that is a power of two, either $\Point = \MPone_{\Point}(\verif)$, or there exists
another power of two $\verif' < \verif$ such that 
$\MPone_{\Point}(\verif) = \MPtwo_{\Point}(\verif')$.
Moreover, since each party is going to roll back to
one of its saved meeting points,
it is not too difficult to verify the following invariant.

\begin{fact}[$\MPone$ is always available]
\label{fact:MPone}
Suppose the simulation has been currently performed up to epoch~$\Point$ (possibly with some previous rolling back).
Then, for any $\verif$ that is a power of two, any positive $\MPone_{\Point}(\verif)$
is a saved meeting point.
\end{fact}

\noindent \textbf{High-Level Intuition of the Simulation Algorithm.}  As mentioned above,
each epoch is supposed to simulate $r$ transmitted bits in the original protocol~$\Pi$.
During simulation, each party has a variable~$\verif$ that keeps track of
how far it is going to roll back the computation.  Each epoch consists of three stages as follows.

\begin{enumerate}

\item \emph{Verification.}  In this stage, for each channel, the parties exchange short hashes and try to
see if they have  the same  $\verif$, the same transcript history, and also whether there is any potential common meeting point.
Two counters~$\voteone$ and $\votetwo$ are used to keep track of the occurrences
of $\MPone$ and $\MPtwo$, respectively.  Observe that rolling back
can occur only when $\verif$ reaches a power of two;
hence, when $\verif$ is a power of 2, there can be potentially $0.5 \verif$ votes for either counter.

Since a party should roll back only if every other party has the same~$\verif$,
each party also has a variable $\error$ that keeps track of the number of epochs with inconsistent~$\verif$
since its $\verif$ is reset to 0.

During the verification, a party first checks if $\verif$ is consistent with the other party.  If not,
then it increases $\error$ and does not check the meeting points;
if $\verif$ is consistent, then it checks which meeting point is common and increases the corresponding vote,
with preference given to $\MPone$ over $\MPtwo$.

\item \emph{Computation.} In case that a party believes all parties have consistent transcript history
and the same variable $\verif = 1$, it will simulate one epoch of the original protocol~$\Pi$.
This increases the transcript history by $r$ bits.

If a party thinks anything is inconsistent or $\verif \neq 1$, then it will perform one dummy epoch simulation.

\item \emph{Transition.}  The main purpose of this stage is for a party to roll back the computation if needed.

If its variable $\error$ is too large, this means that there are too many epochs with inconsistent $\verif$;
in this case, it just resets $\verif$ and $\error$ to 0, and waits for other parties to reset their $\verif$ to start over.

Otherwise, when $\verif$ reaches a power of two and some meeting point has enough votes (at least $0.4 \verif$
out of $0.5 \verif$ possible votes),
then the party will roll back to that meeting point.

Moreover, at the end of the epoch, each party will remove some meeting points to satisfy the
low-memory requirement.

\end{enumerate}

\begin{algorithm}
\caption{Robust Simulation over Oblivious Noisy Channel (Bob version)}
\label{alg:ComputeOblivious}
\begin{small}

\KwIn{Original protocol~$\Pi$, number~$m$ of Bobs, number~$n$ of bits transmitted in $\Pi$, 
error rate $\epsilon$ of communication channel,
failure probability~$\delta \leq \frac{1}{n}$ of simulation}


Let $\lhash$ and $\shash$ be from Definitions \ref{def:transcripthash} and \ref{def:innerproducthash} respectively with
$o=\Theta(\log \frac{1}{\delta})$ and $c=\Theta(\log m)$.

Set:
$r \gets \Theta(\sqrt{\frac{c}{m^5 \epsilon}})$;
$R = \ceil{\frac{n}{r}} + \Theta(m^5 n \epsilon)$;
$I \gets \Theta(\log o)$;
$L \gets \log R + o$.


Initialization: $\Point,\verif,\error,\voteone,\votetwo \gets 0$; 

$\HD \gets \emptyset$
\Comment{$\HD$ stores tuples of the form $(\Point, H, \xi_\Point)$,
where $\HD[\Point]$ returns $H$.}

$\core^* \gets $ \textsc{RobustReceive}($l = \Theta(o \log R r o) = \Theta(\log\frac{1}{\delta}\log n),t=2n\epsilon$) from Alice.
\label{algline:robustrandomexchange}

$\mathcal{S}^* \gets$ \textsc{RandInit}($\core^*, R r o, \rho = 2^{-o}$)

\For{$\ceil{\frac{R}{I}}$ phases}
{

Refresh $cL$-bit $\sseed \gets$ \textsc{RelaxedShareRand} with Alice
as in Definition~\ref{defn:relaxedrand}.

\label{algline:randomexchange}

\For{$I$ epochs}
{

    $\verif \gets \verif + 1$; $\VP \gets 2^{\ceil{\log_2 \verif}}$; $\MPone \gets \VP  \floor{\Point/\VP }$; $\MPtwo \gets \MPone - \VP.$
    \label{line:veri_begin}
    \Comment*[f]{\normalfont Verification Stage: Lines \ref{line:veri_begin} to \ref{line:veri_end}}


	$(G_\verif, G_1, G_2, G_\Point) \gets (\shash(\verif, \lseed), \shash((\MPone,\HD[\MPone]),\sseed), \shash((\MPtwo,\HD[\MPtwo]),\sseed), \shash((\Point,\HD[\Point]),\sseed))$.
	\label{algline:hashingend} 
	
	Send $(G_\verif, G_1, G_2, G_\Point)$ to Alice.
	
	Receive $(G_\verif', G_1',G_2', G_\Point')$ from Alice.

	\If{$G_\verif \neq G_\verif'$\label{algline:ComputeAfterHash}}
    {
		$\error \gets \error + 1$.
    }
	\Else
    {
		\If{$G_{1} \in \{G_{1}',G_{2}'\}$}
        {
			$\voteone \gets \voteone + 1$.
            \label{line:voteone}
        }
		\ElseIf{$G_2 \neq \bot$ and $G_{2} \in \{G_{1}',G_{2}'\}$}
        {
			$\votetwo \gets \votetwo + 1$.
            \label{line:votetwo}\label{line:veri_end}
		}
	}

    \BlankLine
    \BlankLine

    
    
    \label{algline:globalseed}
	
	\If(\Comment*[f]{\normalfont Computation Stage: Lines \ref{line:comp_begin} to \ref{line:comp_end}})
    {$\verif = 1$ and $\error = 0$ and $G_{\Point} = G_{\Point}'$}
    {	\label{line:comp_begin}
        
				
		$\Point \gets \Point + 1$; \label{line:computation}
		
		$\lseed^* \gets$ $\mathcal{S}^*$.\textsc{ExtractBlock}$(\Point, b = ro)$

		
     Simulate epoch-$\Point$ of $\Pi$ (with Alice) to get its $r$-bit transcript  $\sigma_\Point$
		and updated internal state $\xi_\Point$ in $\Pi$.
		
		\Comment{{\normalfont The (imaginary) state tape is overwritten for the $r$ rounds in epoch $\Point$ as in Definition~\ref{defn:robust_sim}.}}

		Compute $H \gets {\lhash}(\HD[\Point-1], \sigma_\Point, \lseed^*)$
		and				
		insert $(\Point, H, \xi_\Point)$ into $\HD$.

        Reset counters: $\verif,\error,\voteone,\votetwo \gets 0$.\label{line:resetcomputation}
	}	
	\Else
    {
		Do 1 dummy epoch simulation with Alice.\label{line:comp_end}
	}

    \BlankLine
    \BlankLine

	\If(\Comment*[f]{\normalfont Transition Stage: Lines \ref{line:tran_begin} to \ref{line:tran_end}})
    {$2 \error \geq \verif$ \label{line:tran_begin}\label{line:transitionerror}}
    {
		Reset counters: $\verif,\error,\voteone,\votetwo \gets 0$.
        \textcolor{red}{($\beta \gets 0$)}\label{line:reseterror}
        
    }
	\ElseIf {$\verif = \VP$ and $\voteone \geq 0.4 \cdot \VP$
    \label{line:transitionrollbackone}}
    {
		Rollback computation: $\Point \gets \MPone$ and restore the corresponding local state.

		
        Reset counters: $\verif,\error,\voteone,\votetwo \gets 0$.
        \textcolor{red}{($\alpha\gets \alpha + 0.5\beta$; $\beta \gets 0$;
        \textbf{if} $\disagree_{AB}=0$ \textbf{then} $\alpha\gets 0$)}\label{line:resetrollbackone}
    }
	\ElseIf {$\verif = \VP$ and $\votetwo \geq 0.4 \cdot \VP$ \label{line:transitionrollbacktwo}}
    {
		Rollback computation: $\Point \gets \MPtwo$ and restore the corresponding local state.
		
		
		Reset counters: $\verif,\error,\voteone,\votetwo \gets 0$.
        \textcolor{red}{($\alpha\gets \alpha + 0.5\beta$; $\beta \gets 0$;
        \textbf{if} $\disagree_{AB}=0$ \textbf{then} $\alpha\gets 0$)}\label{line:resetrollbacktwo}
    }
	\ElseIf{$\verif = \VP$}
    {
		$\voteone,\votetwo \gets 0$.
        \label{line:voteresetdouble}
	}

    For all index $p\not=\Point$, remove any  $(p, \cdot, \cdot)$ in $\HD$
	for which there exists no integer $j\geq0$ such that $p = 2^j  \floor{\Point/2^j}-2^j$.
    \label{line:removepoints}\label{line:tran_end}
}
}

\Return{The state tape for the simulation of $\Pi$ as in Definition~\ref{defn:robust_sim}. \label{algline:ComputeEnd}}
\end{small}
\end{algorithm}

\begin{algorithm}
\caption{Robust Simulation over Oblivious Noisy Channel (Alice version)}
\label{alg:AliComputeOblivious}
\begin{small}

\KwIn{Original protocol~$\Pi$, number~$m$ of Bobs, number~$n$ of bits transmitted in each channel in $\Pi$, 
error rate $\epsilon$ of communication channel,
failure probability~$\delta$ of simulation}


Let $\lhash$ and $\shash$ be from Definition \ref{def:transcripthash} and \ref{def:innerproducthash} respectively with
$o=\Theta(\log \frac{1}{\delta})$ and $c=\Theta(\log m)$.

Set:
$r \gets  \Theta(\sqrt{\frac{c}{m^5 \epsilon}})$;
$R = \ceil{\frac{n}{r}} + \Theta(m^5 n \epsilon)$;
$I \gets \Theta(\log o)$;
$L \gets \log R + o$.

Initialization: $\Point,\verif,\error,\voteone,\votetwo \gets 0$; 

$\HD \gets \emptyset$
\Comment{$\HD$ stores tuples of the form $(\Point, (H_i: i \in [m]), \xi_\Point)$,
where $\HD_i[\Point]$ returns $H_i$.}

Sample $l = \Theta(o \log R r o)$-bit $\core^*$ uniformly at random;
for each $i \in [m]$, \textsc{RobustSend}($\core^* ,t=2n\epsilon$)
to $B_i$.
\label{algline:Alirobustrandomexchange}

$\mathcal{S}^* \gets$ \textsc{RandInit}($\core^*, R r o, \rho = 2^{-o}$)

\For{$\ceil{\frac{R}{I}}$ phases}
{

Refresh $cL$-bit $\sseed \gets$ \textsc{RelaxedShareRand} with every Bob
as in Definition~\ref{defn:relaxedrand}.
\label{algline:Alirandomexchange}

\For{$I$ epochs}
{
    $\verif \gets \verif + 1$; $\VP \gets 2^{\ceil{\log_2 \verif}}$; $\MPone \gets \VP  \floor{\Point/\VP}$; $\MPtwo \gets \MPone - \VP.$
    \label{line:Aliveri_begin}
    \Comment*[f]{\normalfont Verification Stage: Lines \ref{line:Aliveri_begin} to \ref{line:Aliveri_end}}

    Compute $G_\verif \gets \shash(\verif, \lseed)$. \label{algline:Alilocalseed}



    $\forall i\in[m]:$ Receive $(G_{\verif,i}', G_{1,i}',G_{2,i}', G_{\Point,i}')$ from $B_i$.

	\If{$\exists i\in [m]: G_\verif \neq G_{\verif,i}'$\label{algline:AliComputeAfterHash}}
    {

        $\forall i\in[m]:$ Send $(\bot, \bot, \bot, \bot)$ to $B_i$.

		$\error \gets \error + 1$.
    }
	\Else
    {
        \For{\emph{each} $i\in [m]$}
        {
            Compute $(G_{1,i},G_{2,i},G_{\Point,i}) \gets ( \shash((\MPone,\HD_i[\MPone]),\sseed), 
						\shash((\MPtwo,\HD_i[\MPtwo]),\sseed),
            \shash((\Point,\HD_i[\Point]),\sseed)).$
        }
        \If{$\exists i\in [m]: G_{\Point,i}\not=G_{\Point,i}'$}
        {
            $\forall i\in [m]: G_{\Point,i} \gets \bot$.
        }

		\If{$\forall i\in [m]: G_{1,i} \in \{G_{1,i}',G_{2,i}'\}$}
        {
			$\voteone \gets \voteone + 1$.
            \label{line:Alivoteone}
        }
		\ElseIf{$G_2 \neq \bot$ \emph{and} $\forall i\in [m]: G_{2,i} \in \{G_{1,i}',G_{2,i}'\}$}
        {
			$\votetwo \gets \votetwo + 1$.
            \label{line:Alivotetwo}

            $\forall i\in [m]: G_{1,i}\gets\bot$.
		}
        \Else
        {
            $\forall i\in [m]: (G_{1,i},G_{2,i})\gets(\bot,\bot)$.
        }

        $\forall i\in[m]:$ Send $(G_\verif, G_{1,i},G_{2,i}, G_{\Point,i})$ to $B_i$.\label{line:Aliveri_end}

	}	
	

    
    \BlankLine
    \BlankLine
	
	\If(\Comment*[f]{\normalfont Computation Stage: Lines \ref{line:Alicomp_begin} to \ref{line:Alicomp_end}})
    {$\verif = 1$ and $\error = 0$ and $\forall i\in [m]: G_{\Point,i} = G_{\Point,i}'$ }
    {	\label{line:Alicomp_begin}
        
				
		$\Point \gets \Point + 1$;

		$\lseed^* \gets$ $\mathcal{S}^*$.\textsc{ExtractBlock}$(\Point, b = ro)$
		\label{algline:Aliglobalseed}
		
		
    Simulate epoch $\Point$ of $\Pi$ (with all $m$ Bobs) to get its $r$-bit transcript $\sigma_{\Point, i}$
		from each channel for $i \in [m]$,
		and updated internal state $\xi_\Point$ in $\Pi$.
		
		\Comment{{\normalfont The (imaginary) state tape for each channel is overwritten for the $r$ rounds in epoch $\Point$ as in Definition~\ref{defn:robust_sim}.}}

		For each $i \in [m]$, compute $H_i \gets {\lhash}(\HD_i[\Point-1], \sigma_{\Point,i}, \lseed^*)$.
		
		Insert $(\Point, (H_i: i \in [m]), \xi_\Point)$ into $\HD$.

				Reset counters: $\verif,\error,\voteone,\votetwo \gets 0$.\label{line:Aliresetcomputation}\label{line:Alicomputation}

	}	
	\Else
    {
		Do 1 dummy epoch simulation of $r$ rounds with all $m$ Bobs.\label{line:Alicomp_end}	
	}

    \BlankLine
    \BlankLine

    {\bfseries The same code as Lines \ref{line:tran_begin} to \ref{line:tran_end} in Algorithm \ref{alg:ComputeOblivious}.}
    \label{line:Alitran_begin}
%
		%
		\label{line:Alitran_end}
		
}
}

\Return{The state tapes for the simulation of $\Pi$ as in Definition~\ref{defn:robust_sim}.}
\end{small}
\end{algorithm}

From the description in Algorithms~\ref{alg:ComputeOblivious}
and~\ref{alg:AliComputeOblivious},
it is straightforward to analyze the simulation overhead and memory usage.

\begin{lemma}[Simulation Overhead]
\label{lemma:overhea}
The communication overhead of Algorithms~\ref{alg:ComputeOblivious}
and~\ref{alg:AliComputeOblivious} is at most
$1 + O(m^{2.5} \sqrt{\log m}) \cdot \sqrt{\epsilon} \leq 2$.
\end{lemma}

\begin{proof}
It suffices to consider Algorithm~\ref{alg:ComputeOblivious}
from one Bob's perspective.

Transmitting $\core^*$ robustly takes $O(\log \frac{1}{\delta} \log n + n \epsilon) = O(n \epsilon)$,
by Lemma~\ref{lem:numOfBitsRobust}.

The information transmitted in each epoch of $\Pi'$ consists of the following:
\begin{compactitem}

\item The simulation of one epoch in the original $\Pi$, which takes $r$ bits.

\item There are 8 short hashes which takes $8c$ bits in total, because each short hash is $c$-bit long.

\item For every $I$ epochs in a phase, each call of
the subroutine \textsc{RelaxedShareRand} in Definition~\ref{defn:relaxedrand}
takes $O(c \log c L + I) = O(c \log o + I)$ bits.
\end{compactitem}

Hence, the total number of bits transmitted is

$O(n \epsilon) + R \cdot \{r + 8c + O(\frac{c \log o}{I}) \}
\leq n + n \cdot O(m^5 \epsilon r + \frac{c}{r} + m^5 \epsilon c).
$

Choosing $r = \Theta(\sqrt{\frac{c}{m^5 \epsilon}}) = \Theta(\sqrt{\frac{\log m}{m^5 \epsilon}})$
and recalling that $m \leq O(\frac{1}{\epsilon})^{0.199}$,
the overhead is at most $1 + O(m^{2.5} \sqrt{\log m}) \cdot \sqrt{\epsilon} \leq 2$.
\ignore{
**********
The total round complexity of the main loop Algorithm \ref{alg:ComputeOblivious} is
\begin{align*}
mR (r + r_c) &= m[\ceil{n / r} + (2mC_6+2C_6+m)\Theta(mn \eps)] r (1 + \frac{r_c}{r})\\
&= mn [1 + \Theta(m^4r \eps)] (1 + \frac{r_c}{r})
= mn [1 + \Theta(m^4\sqrt{\eps\log m})].
\end{align*}
where by Algorithm \ref{alg:ComputeOblivious},  $r = \ceil{\sqrt{\frac{r_c}{\eps}}}$ and $r_c = \Theta(\log m)$.
Moreover, the communication performed by the robust randomness agreement is $\Theta(mn \sqrt{\eps})$ many rounds and therefore negligible. 
Thus, the communication rate of Algorithm \ref{alg:ComputeOblivious} is $1 - \Theta(m^4\sqrt{\eps\log m})$ as desired. 
}
\end{proof}

\ignore{
\begin{theorem}
[Theorem 7.1]
\label{the:mainOblivious}
Suppose any $n$-bit protocol $\Pi$ using any alphabet $\Sigma$. Algorithm \ref{alg:ComputeOblivious} is a computationally efficient randomized coding scheme which given $\Pi$, with probability at least $1 - \delta$, robustly simulates it over any oblivious adversarial error channel with alphabet $\Sigma$ and error rate $\eps$. The simulation uses $mn (1 + \Theta(m^4\sqrt{\eps\log m}))$ rounds and therefore achieves a communication rate of $1 - \Theta(m^4\sqrt{\eps\log m})$.
\end{theorem}
}

\begin{lemma}[Memory Usage]
Suppose the memory usage of Alice in the original protocol is $M_A$.
Then, the memory usage of Alice in the simulation in Algorithm~\ref{alg:AliComputeOblivious}
is $O(m \log \frac{1}{\delta} \log n) + O(\log n) \cdot M_A$.

The memory usage of each Bob is similar, as if $m = 1$.
\end{lemma}

\begin{proof}
For processing the pre-shared randomness (for both long and short hashes),
the memory usage is dominated by handling $\core^*$, which takes $O(o \log n) = O(\log \frac{1}{\delta} \log n)$ bits.

Since Alice is involved in $m$ communication channels, each of the $O(\log n)$ meeting  point takes $O(m o)$ bits for
the $m$ long hashes, plus $M_A$ bits for the corresponding internal state in $\Pi$.

Therefore, the total memory usage is $O(m \log \frac{1}{\delta} \log n) + O(\log n) \cdot M_A$.
\end{proof}

\ignore{

\begin{table*}[!ht]
\caption{Algorithm Parameters}
\label{tab:AlgPara}
\centering
\begin{tabular}{cl}
\toprule
    Parameter & Description \\
\midrule
    $\epsilon$ & corruption rate  \\
    $r_c$ & size of hash in bits  \\
    $r$ & number of simulated exchanged bits per epoch  \\
    $R$ & total number of epochs \\
    $I$ & number of epochs per phase \\
    $L$ & number of bits to index an epoch \\
    $\sigma_p$ & information exchanged in epoch~$p$ \\
    $\HD$ & hash dictionary \\
\bottomrule
\end{tabular}
\end{table*} 
}

%% file: proofmulti.tex
\section{Correctness of Simulation}
\label{sec:proofmultiBob}

\ignore{
\textcolor{red}{(New feature: Extend the algorithm to multiple Bobs.)}

Parameters used in Algorithm \ref{alg:ComputeOblivious} are listed in Table \ref{tab:AlgPara}.
The dictionary $\HD$ is supposed to store epoch-hash pairs $(p, H_p)$.
We use the convention that $\HD[p]$ returns $H_p$ if the pair is stored;
otherwise, $\HD[p]$ returns $\bot$.  Moreover, for the case that $(p, H_p)$ is stored in $\HD$,
we also assume implicitly that the party also stores its internal state after epoch~$p$ is simulated.
}

While the proofs for communication overhead and memory usage are quite straightforward
from the parameters in Algorithms~\ref{alg:ComputeOblivious}
and~\ref{alg:AliComputeOblivious},
the correctness of simulation is more subtle.  The following can be potential correctness issues:

\begin{compactitem}
\item After $R$ epochs of simulation, the original protocol $\Pi$ still has not reached the end.

\item The original protocol $\Pi$ has reached the end, but the transcripts of different parties are inconsistent
because of hash collision and/or corruption.
\end{compactitem}

\noindent \textbf{Proof Outline.}
We adapt the potential function argument in~\cite{Haeupler14,abs-1805-06872} to resolve the above issues.  The novel parts
of our analysis are as follows.
\begin{compactitem}
\item We consider multiple number~$m$ of Bobs.  Hence, the constants in the original potential function
now all have dependence on $m$, which will eventually show up in the communication overhead.

\item Similar to~\cite{abs-1805-06872}, we introduce extra variables to accommodate the case when meeting point $\MPtwo$ of some scale
is unavailable, and make some of the analysis more explicit.
\end{compactitem}

In view of Corollary~\ref{cor:long_hash_len},
the probability of long hash collision is small.  Hence,
we can use the same proof structure as in~\cite{Haeupler14}.
A potential function $\Phi$ is defined for the following proof strategy.

\begin{enumerate}

\item If a corruption or short hash collision does not happen in an epoch, the potential function increases by some amount.

\item If a corruption or short hash collision happens in an epoch, the potential function decreases by at most some amount.

\item Except with small probability, the number of short hash collisions is small and comparable to the number of corruptions.

\item Therefore, if the simulation is run with a sufficient number of epochs,
the final potential function will have a large value, except with small probability.
The potential $\Phi$ is designed such that this implies that the simulation of the original protocol is finished.

\end{enumerate}



\noindent \textbf{Potential Function.}  
We use a similar potential function $\Phi$ as in~\cite{Haeupler14,abs-1805-06872},
although some variables are renamed or decomposed for clarity;
it is defined with respect to
 the variables in 
Algorithms~\ref{alg:ComputeOblivious}
and~\ref{alg:AliComputeOblivious}, together with some 
additional auxiliary variables.
Recall that $A$ denotes Alice and
$B_i$ denotes each Bob.
We use $AB = \{A, B_1, \ldots, B_m\}$ to denote all parties;
when $AB$ is used as a subscript for a variable,
it means the summation of that variable over all parties.
For example, 
$\verif_{AB} = \sum_{i\in AB} \verif_{i}$.


\noindent \emph{Transcript Prefix.}  Although
each party does not remember its complete transcript history,
we use $\sigma[1..\Point]$ to denote the transcript prefix
that leads to that party's current internal state 
in the simulation of $\Pi$.  In Alice's case,
we use $\sigma_{A,i}$ to denote the transcript corresponding
to communication with $B_i$.
Recall that because of Corollary~\ref{cor:long_hash_len},
we can assume that there is no long hash collision in the potential analysis.

\noindent \emph{Maximal Common Prefix.}
We define the maximal agreement parameter $\agree$ as 
$$\agree :=  \max \left\{l \in [1..\min\{\Point_A,\Point_{B_1},...,\Point_{B_m}\}] \ \text{s.t.}\  \forall i\in[m]: \sigma_{A,i}[1..l] = \sigma_{B_i}[1..l]\right\}.$$
We also define the disagreement parameter
 $\disagree_j :=\Point_j-\agree$, for each party $j \in AB$.

\begin{remark}
Observe that if $\agree$ reaches $\ceil{\frac{n}{r}}$,
then all parties have consistent transcripts of at least $n$ bits in the original protocol $\Pi$.
This means that the original protocol must have been completed,
and the corresponding state tapes will satisfy Definition~\ref{defn:robust_sim}.
\end{remark}

\noindent \emph{Bad Vote Counters.}  These are known as $\mathsf{BVC}$
in~\cite{Haeupler14,abs-1805-06872}.
For each party $j \in AB$, there is a variable $\beta_j$;
these variables keep track of short hash collisions and corruptions
related to meeting points.
Each of them can increase by at most one in each epoch;
moreover, when one of them increases by one, then \textbf{all of them}\footnote{
Increasing all $\beta$ variables together can simplify the proof.  Considering
each channel separately might get a better dependence on $m$, but we do not
see a simple way to achieve so.} 
must increase by one,
if at least one of the following happens during an epoch:

\begin{enumerate}
\item There exists some party~$j$ such that $\voteone$ of that party 
increases (Line~\ref{line:voteone} in Algorithm \ref{alg:ComputeOblivious} and Line~\ref{line:Alivoteone} in Algorithm \ref{alg:AliComputeOblivious}),
but its transcript $\sigma_j[1..\MPone]$ in some channel
actually does not match 
$\sigma_{j'}[1..\MPone]$ or $\sigma_{j'}[1..\MPtwo]$
of the other party~$j'$.
This happens due to short hash collision or corruption.

\item Similarly, 
there exists some party~$j$ such that its $\votetwo$ is increased due to short hash collision or corruption.


\item There exists some party whose $\voteone$ or $\votetwo$ does not increase because 
the short hash for some meeting point is corrupted during transmission; 
in other words, if there were no corruption for the short hashes of meeting points during transmission in that epoch,
then its $\voteone$ or $\votetwo$ would have increased.
\end{enumerate}

The $\beta$ variables for different parties will be decreased separately.
When a party $j \in AB$
resets every counter
(Line~\ref{line:reseterror}, \ref{line:resetrollbackone} or \ref{line:resetrollbacktwo} in Algorithm \ref{alg:ComputeOblivious}), we reset $\beta_j$ of that party to zero.

\noindent \textbf{Auxiliary Variables for Missing Meeting Points.}
In~\cite{abs-1805-06872}, a variable $L^-$ is defined to accommodate the analysis
of missing meeting points.  However, we find that the analysis is clearer, if
we consider two variables $\alpha$ and $\gamma$ as follows.

\noindent \emph{Bad Vote Accumulator.} These variables act as a buffer to 
accumulate values from the $\beta$ variables.
For each party $j \in AB$, there is a variable $\alpha_j$, which is modified in the following ways:
\begin{enumerate}
\item In Line~\ref{line:resetrollbackone} and \ref{line:resetrollbacktwo} in Algorithm \ref{alg:ComputeOblivious},
just before a party~$j$ resets $\beta_j$ to zero, 
the value of its $\alpha_j$ increases by half of~$\beta_j$.

\item  When a party $j$ does a meeting point transition
(Line~\ref{line:resetrollbackone} and \ref{line:resetrollbacktwo} in Algorithm \ref{alg:ComputeOblivious}) with $\disagree_{AB}=0$ after the transition, its  $\alpha_j$ is reset to zero.

\end{enumerate}

\noindent \emph{Corrupted Computation.}  We have a single variable $\gamma$
that keeps track of
simulation of $\Pi$ that is corrupted; it is modified as follows.

\begin{enumerate}

\item  During the computation stage of an epoch,
if $\exists j\in AB: \disagree_j$ increases by one, then we call this epoch suffers a corrupted computation and $\gamma$ is increased by one.

Note that the increase of $\disagree_i$ of a party $i$ caused by rollback does not count as a corrupted computation and $\gamma$ does not increase.

\item Suppose all parties do meeting point transitions with $\forall i\in AB: 0<2\disagree_{AB}<\verif_{i}=\verif$ and $\alpha_{AB}+\beta_{AB} < 0.1\verif$ before the transition and with $\disagree_{AB}=0$ after the transition.
In this case, the $\gamma$ value will decrease by $0.25\verif$.  As we shall see in Lemma~\ref{lem:gammaDecrease},
when $\gamma$ is decreased, it will never drop below 0.

\ignore{
\item
Let $\Point_{\max}$ be the maximum value that $\Point_i$ of any party $i$ has ever reached.
If all parties do one epoch simulation in the computation stage
and $\agree$ increases to $\Point_{\max}$,
then we will reset $\gamma=0$ after the transition stage (during which all parties will do error transition since $\verif=\error=0$).
}

\end{enumerate}

\noindent \textbf{Intuition for the New Variables $\alpha$ and $\gamma$.}
We explain briefly how these variables are used to handle the case when some $\MPtwo$ is missing.
Consider some index $p = (4n+1) 2^i$, which is removed from $\HD$ when
a party~$j$ reaches $\Point = p + 2^{i+1}$.  There are two scenarios in which
the party needs to consider $p$ as a potential meeting point, where we will
need either $\alpha_{AB}$ or $\gamma$ to be large for the potential analysis.
If $\agree$ is close to $\Point$, then this means that
to roll so further back to $p$, many available meeting points are overlooked
due to bad votes, which is accumulated by $\alpha_{AB}$;
if $\agree$ is not close to $\Point$, then this means that
there must be many corrupted computations by party~$j$,
which is counted by~$\gamma$.

We remark that the parties are not aware of the potential function,
whose purpose is to analyze the correctness of the simulation.
Now we can define the potential to be %

\begin{numcases}{\Phi=}
\agree - C_3  \disagree_{AB} + C_2  \verif_{AB}  - C_5  \error_{AB} - C_6 \alpha_{AB}-2C_6\beta_{AB}- C_7\gamma, \, \quad
\text{if } \forall i,j\in AB: \verif_i = \verif_{j}; \label{eq:potential_eq}\\
\agree - C_3 \disagree_{AB} - 0.9 C_4  \verif_{AB}  +  C_4  \error_{AB} -  C_6  \alpha_{AB}- C_6\beta_{AB}- C_7\gamma, \, \,
\text{otherwise.} \label{eq:potential_neq}
\end{numcases}
where the coefficients $C_2$ to $C_7$ are listed in Table \ref{tab:Constant}.

\begin{table*}[!ht]
\caption{Coefficients in Potential $\Phi$}
\label{tab:Constant}
\centering
\begin{tabular}{cl}
\toprule
    Constants & Value \\
\midrule
    $C_2$ & $1$  \\
    $C_3$ & $6 + 2m$ \\
    $C_4$ & $30 + 60m + 20 m^2$  \\
    $C_5$ & $6 + 80m + 116m^2 + 36m^3$ \\
    $C_6$ & $310m + 180m^2 + 18m^3$ \\
    $C_7$ & $12 + 4m$ \\
\bottomrule
\end{tabular}
\end{table*} 

%
%

The proof strategy can be summarized in the following technical lemma,
whose proof is deferred to Section~\ref{sec:potential_tech}.

\begin{lemma}[Technical Lemma for Potential Function]
\label{lemma:potential_tech}
The following statements are true for the potential function $\Phi$.
\begin{enumerate}

\item If there is no corruption or hash collision in an epoch,
the potential $\Phi$ increases by at least 1;
otherwise, the potential $\Phi$ decreases by at most $O(m^4)$.

\item Except with probability $\exp(-\Theta(\frac{n^{0.25} \epsilon}{I}))$,
the number of epochs with short hash collision is at most $O(m n \epsilon)$.

\item When the simulation terminates, the agreement parameter $\agree \geq \Phi$.
\end{enumerate}
\end{lemma}

\begin{lemma}[Correctness of Simulation]
\label{lemma:correct_sim}
Suppose the failure probability satisfies $\exp(-\Theta(\frac{n^{0.25}  \epsilon}{\log n\epsilon})) \leq \delta \leq \frac{1}{n}$.  
Then, except with probability~$\delta$,
Algorithms~\ref{alg:ComputeOblivious}
and~\ref{alg:AliComputeOblivious}
correctly simulate the protocol $\Pi$.

\end{lemma}

\begin{proof}
By Corollary~\ref{cor:long_hash_len},
by choosing the length of the long hash to be $o = \Theta(\log \frac{1}{\delta})$,
the probability of a long hash collision is at most $\frac{\delta}{2}$.

From the second statement of Lemma~\ref{lemma:potential_tech},
the probability of having more than $\Theta(mn \epsilon)$ epochs with short
hash collision is at most $\exp(-\Theta(\frac{n^{0.25} \epsilon}{I})) \leq
\exp(-\Theta(\frac{n^{0.25} \epsilon }{\log \log \frac{1}{\delta}}))$,
since $I = O(\log o)$.  Since we have $\exp(-\Theta(\frac{n^{0.25}  \epsilon}{\log n\epsilon})) \leq \delta$,
it follows that the constant in the big-theta can be chosen such that
$\exp(-\Theta(\frac{ n^{0.25} \epsilon }{\log \log \frac{1}{\delta}})) \leq \frac{\delta}{2}$.

Therefore, by the union bound, except with probability $\delta$,
there is no long hash collision and the number of epochs with short hash collisions
is at most $O(mn \epsilon)$.
Moreover, 
since the overhead is at most 2 from Lemma~\ref{lemma:overhea},
the $\epsilon$ error rate of the $m$ channels
implies that there can be at most $O(m n \epsilon)$ epochs with corruption.

From the first statement of Lemma~\ref{lemma:potential_tech},
by having a large enough big-theta constant in $R = \ceil{\frac{n}{r}} + \Theta(m^5 n \epsilon)$,
this implies that finally the potential $\Phi$ is at least $\ceil{\frac{n}{r}}$.

By the third statement of Lemma~\ref{lemma:potential_tech},
all parties have an agreed simulation of at least $\ceil{\frac{n}{r}}$ epochs
of $\Pi$, which means the simulation of $\Pi$ is completed.
\end{proof}

%% file: potential_tech.tex
\section{Technical Proofs for Potential Function}
\label{sec:potential_tech}

This section gives the detailed proof of Lemma~\ref{lemma:potential_tech},
whose results are adapted from~\cite{Haeupler14,abs-1805-06872} to multiple
number of Bobs.  For completeness, we give the full proof here, and emphasize
which parts are novel.

\subsection{Lower Bound on Potential Function}
\label{sec:potential_lb}

This is the first statement of Lemma~\ref{lemma:potential_tech},
which states how fast the potential function grows during the simulation.

In
Algorithms~\ref{alg:ComputeOblivious}
and~\ref{alg:AliComputeOblivious},
an epoch is
\emph{consistent} if $\disagree_{AB}=0$ and $\forall i,j\in AB: \verif_{i}=\verif_{j}$ at the beginning of the epoch; otherwise, the epoch is \emph{inconsistent}.

As in~\cite{Haeupler14}, we decompose each epoch into three stages:
verification, computation and transition.
The first lemma analyzes how the potential function changes
when the algorithm goes through the verification and computation stages.

\begin{lemma}[Verification and Computation Stages: Similar to
Lemma 7.3 in~\cite{Haeupler14}]
\label{lem:potentialIncreaseVCPhase}
Fix some epoch and suppose all parties run their verification and computation stages (e.g., Line \ref{line:veri_begin} to \ref{line:comp_end} in Algorithm~\ref{alg:ComputeOblivious} and Line \ref{line:Aliveri_begin} to \ref{line:Alicomp_end} in Algorithm~\ref{alg:AliComputeOblivious}).
If there exists at least one corruption or short hash collision,
the potential decreases by at most $O(m^4)$.
Furthermore, if the corresponding epoch is consistent and no corruption occurs, the potential increases by at least one.
If the corresponding epoch is inconsistent and no corruption or short hash collision occurs, the potential increases by at least two.
\end{lemma}
\begin{proof}
We consider the following two cases:
\begin{enumerate}
\item
First, we consider the case when there exists at least one corruption or short hash collision.
We use the notations with a superscript $*$ to denote the values before the verification stage
and those without a superscript to denote the values after the computation stage.
We also use a $\Delta$ in front of any variable to denote the change of value to this variable.
During the verification stage and the computation stage,
$\alpha_{AB}$ does not change and thus we can ignore $\alpha_{AB}$.
Let $\sset_2$, $\sset_1$ and $\sset_0$ denote the set of parties in $AB$ with $\verif\geq2$, $\verif=1$ and $\verif=0$, respectively,
at the end of the computation stage.

Then, we consider the following scenarios:

\begin{enumerate}
\item $\exists i, j\in AB: \verif^*_i \not= \verif^*_j$ and 
$\forall i,j \in AB: \verif_i = \verif_j$.

Observe that for each party, the variable $\verif$ is incremented by one during the verification stage;
moreover, the variable will be reset to 0 only if $\verif^*_i$ is initially 0.

Hence, it follows that there exist $i \neq j \in AB: \verif_i \neq \verif_j$.
Therefore, this case is impossible.

\item
\label{case:neq2neq}
$\exists i,j\in AB: \verif^*_i \not= \verif^*_j$
and $\exists i,j\in AB: \verif_i \not= \verif_j$.

Then, 
both the potential before the verification stage and after the computation stage can be evaluated by \eqref{eq:potential_neq}. 

For each party $i\in S_0$,
$i$ should do one epoch simulation in the computation stage;
thus, the potential decreases the most
when $\Delta\disagree_i,\Delta\beta_i,\Delta\gamma=1$.
For each party $i\in S_1\cup S_2$,
$i$ does 1 dummy epoch simulation in the computation stage
and thus the potential decreases the most
when $\Delta\verif_i,\Delta\beta_i=1$.

Therefore,
the decrease of the potential is at most
$0.9C_4(|S_1|+|S_2|)+(m+1)C_6+C_3|S_0|+C_7 = O(m^4)$.

\item
\label{case:eq2neq}
$\forall i,j\in AB: \verif^*_i = \verif^*_j$
and $\exists i,j\in AB: \verif_i \not= \verif_j$.

Observe that if the common value $\verif^*_i > 0$,
then it follows that they must have the common value $\verif_i = \verif^*_i  + 1 \geq 2$
at the end of the computation stage.
Hence, we must have $\forall i\in AB: \verif^*_i = 0$,
as well as $\error^*_i = \beta^*_i = 0$.

Moreover, in order to have $\verif_i \not= \verif_j$,
there must be at least one party that does actual epoch simulation
and at least one party that does dummy simulation in the computation stage.
Then, from \eqref{eq:potential_eq} the potential before the verification stage is
$$
\Phi^*=
\agree^* - C_3  \disagree_{AB}^* + C_2  \verif_{AB}^*  - C_5  \error_{AB}^* -2C_6\beta_{AB}^* - C_7\gamma^*
$$
and from \eqref{eq:potential_neq} the potential after the computation stage is
$$
\Phi=
\agree - C_3 \disagree_{AB} - 0.9 C_4  \verif_{AB}  +  C_4  \error_{AB} - C_6\beta_{AB}- C_7\gamma.
$$

Recall that $\sset_1$ and $\sset_0$ denote the set of parties in $AB$ with $\verif=1$ and $\verif=0$, respectively;
in this case, $\sset_2$ is empty.

For each party $i\in S_0$,
$i$ does actual simulation in the computation stage
and thus the potential decreases the most
when $\Delta\disagree_i, \Delta\beta_i,\Delta\gamma=1$.
For each party $i\in S_1$,
$i$ should do dummy epoch communication in the computation stage
and thus the potential decreases the most
when $\Delta\verif_i,\Delta\beta_i=1$.
Therefore,
the decrease of the potential is at most
$0.9C_4|S_1| + (m+1)C_6 + C_3|S_0| + C_7 = O(m^4)$.

\item
$\forall i,j\in AB: \verif^*_i = \verif^*_j$
and $\forall i,j\in AB: \verif_i = \verif_j$.

Then, the parties should either all do one actual epoch simulation
or all do one epoch dummy simulation in the computation stage.
Moreover, both the potential before the verification stage and after the computation stage can be evaluated by \eqref{eq:potential_eq}.

If all parties do one actual epoch simulation,
then we should have $\forall i\in AB: \verif^*_i = \verif_i = \error^*_i = \error_i = \beta^*_i=0$.
The worst case is for all parties~$i$, $\Delta \disagree_i = \Delta \beta_i = \Delta \gamma = 1$.
Therefore,
the decrease of the potential is at most
$(m+1)C_3 + (m+1) C_6 + C_7 = O(m^4)$.

If all parties do dummy simulation,
then we should have $\forall i\in AB: \Delta\verif_i = 1$.
Since $C_5<2C_6$ and only at most one of $\error_i$ and $\beta_i$ can increase,
the potential decreases the most
when $\Delta\beta_i=1$.
Therefore,
the decrease of the potential is at most
$(m+1)(2C_6 - C_2) = O(m^4)$.
\end{enumerate}
Therefore, the potential decreases at most $O(m^4)$.

\item
If no corruption or short hash collision occurs,
then the $\disagree_{AB}$, $\alpha_{AB}$, $\beta_{AB}$ and $\gamma$ values do not change.
Moreover, either all parties perform an actual epoch simulation or
all perform a dummy epoch simulation.
Thus, we can ignore the corresponding parts in the potential.
Then, we only need to consider the following three scenarios:
\begin{enumerate}

\item All parties do one actual epoch simulation. 

Then, we must have $\forall i\in AB: \verif^*_i = \verif_i =  \error^*_i = \error_i = 0$.
Moreover, since there is no corruption or hash collision,
it means $\disagree_{AB}=0$.

Hence, $\agree $ increases by one during the computation stage, which implies the potential increases by one.
Note that the epoch is consistent in this scenario.

\item All parties perform one dummy epoch simulation and $\forall i,j\in AB: \verif_i = \verif_j$.

Observe that for all parties~$i$, $\verif_i = \verif^*_i + 1$,
and $\error_{AB}$ and $\agree $ do not change.
Thus, the potential in \eqref{eq:potential_eq} increases by $(m+1)C_2\geq 2$ for $C_2\geq \frac{2}{m+1}$.
Note that the epoch can be either consistent or inconsistent in this scenario.

\item All parties do one dummy epoch communications and $\exists i,j\in AB: \verif_i \not= \verif_{j}$.

Then, $\Delta \verif_i = \Delta \error_i = 1$ for each $i\in AB$, while $\agree $ does not change.
Thus, the potential in \eqref{eq:potential_neq} increases $(m+1)(-0.9C_4+C_4)\geq 2$ for $C_4\geq \frac{20}{m+1}$.
Note that the epoch is inconsistent in this scenario.
\end{enumerate}
Therefore, the potential increases as required with the $C_2$ and $C_4$ listed in Table \ref{tab:Constant}.
\end{enumerate}
After considering the above two cases, we complete the proof for the statement.
\end{proof}

\noindent \textbf{Transition Stage.}
We define some notions relevant to the transition stage. 
We call a transition due to Line~\ref{line:transitionerror} an \textbf{error transition}, a transition due to Line~\ref{line:transitionrollbackone} or Line~\ref{line:transitionrollbacktwo} a \textbf{meeting point transition}
and any of the remaining transitions a \textbf{normal transition}.

For a party~$i \in AB$ at position $\Point_i$, its set of meeting points
is defined to be
 
$\MP_i := \{ 2^j \cdot \floor{\frac{\Point_i}{2^j}} - 2^j \geq 0: j \geq 0 \} \cup \{\Point_i\}$.

Recall that because of rolling back,
a party~$i$ might not have saved all its meeting points in $\MP_i$,
in which case a meeting point is called \emph{missing} or \emph{unavailable}.
An index $p$ is a \textbf{true} meeting point if $p \in \cap_{i \in AB} \MP_i$
and $p \leq \agree$;
otherwise, it is a \textbf{false} meeting point.

%
%
Before we analyze what happens to the potential function during the transition stage,
we establish some new properties, some of which are relevant to multiple number of Bobs.

\begin{lemma}[Beginning of Transition Stage]\label{lem:ELessThanK}
At the beginning of the transition stage,  a party must have {\normalfont$\error\leq 0.5(\verif+1)$}.
\end{lemma}

\begin{proof}
We consider the following two cases:
\begin{enumerate}
\item If the party do one epoch communication right before the transition stage, then both $\verif$ and $\error$ are reset to 0
and we have $\error\leq 0.5(\verif+1)$.

\item 
Otherwise, by Algorithms \ref{alg:ComputeOblivious} and \ref{alg:AliComputeOblivious}, at the end of each epoch, we have $2\error'<\verif'$ or $\error'=\verif'=0$.
In the next epoch before the transition stage,
$\verif'$ will increase by one and $\error'$ will increase by one or stay the same.

If $2\error'<\verif'$, then we have 
$2(\error-1)\leq2\error'<\verif'=\verif-1$, i.e. $\error< 0.5(\verif+1)$ at the beginning of the transition stage.

If $\error'=\verif'=0$, then we have 
$\verif=1$ and 
$\error=0$ or 1 at the beginning of the transition stage.
Thus, we have $\error\leq 0.5(\verif+1)$ again. 
\end{enumerate}
Therefore, we complete the proof for the statement.
\end{proof}

\begin{lemma}[Meeting Point Transitions]\label{lem:changeOfAgree}
For any non-empty subset $\sset\subseteq AB$,
suppose only the parties in $\sset$ do meeting point transitions. 
Then, during the transition stage, 
the decrease of $\agree$ is at most
$2\max_{i\in\sset}\verif_{i}-1$
and the increase of $\disagree_{AB}$ is at most
$m(2\max_{i\in\sset}\verif_{i}-1)$.
\end{lemma}

\begin{proof}
Observe that the increase of $\disagree_{AB}$ can only be at most $m$ times the decrease of $\agree$ when rollback happens.
Thus, it suffices to find the upper bound of the decrease of $\agree$.

Since only the parties in $\sset$ do meeting point transitions,
each party $i \in \sset$ can roll back at most $2\verif_i-1$.
Thus, $\agree$ can only decrease by at most $2\max_{i\in\sset}\verif_{i}-1$,
which happens when $\agree$ equals the $\Point$ of the party in $\sset$ with maximal $\verif$.
In addition, the increase of $\disagree_{AB}$ is at most
$m(2\max_{i\in\sset}\verif_{i}-1)$.
\end{proof}

\begin{lemma}
[True Meeting Point]
\label{lem:haveATrueMP}
Suppose $\verif=2^u$ for some non-negative integer $u$.
If $\forall i\in AB: \verif_i=\verif> \disagree_{AB}$,
then there is a true meeting point that is a multiple of $\verif$.
\end{lemma}

\ignore{
\begin{figure}[h]
    \centering
    \includegraphics[width=2.5in]{mp_case.pdf}
    \caption{Example.}
    \label{fig:mp_case}
\end{figure}
}

\begin{proof}
Recall that  $\Point_i$ is the current position of party $i\in AB$.
Let $w$ be a non-negative integer such that $w\verif \leq\agree <(w+1)\verif$.
Then, we have $\agree +\disagree_{AB} <(w+2)\verif$,
and $\forall i\in AB: w\verif \leq \Point_i <(w+2)\verif$.
Let $\sset_1$ denote the set of parties $i$
such that $w\verif \leq \Point_i <(w+1)\verif$
and $\sset_2$ denote the set of parties $i$
such that $(w+1)\verif \leq \Point_i <(w+2)\verif$.
Then, for each party $i\in \sset_1$,
the party $i$ should have $\MPone=w\verif$
and $\MPtwo=(w-1)\verif$.
For each party $i\in \sset_2$,
the party $i$ should have $\MPone=(w+1)\verif$
and $\MPtwo=w\verif$.
Hence, all parties should have a true meeting point $w\verif$.
\end{proof}

\begin{lemma}
[Removing Meeting Points]
\label{lem:missMP}
Suppose $p=(2w+1)2^u$ for some non-negative integers $w$ and $u$. Then, we have:
\begin{enumerate}
\item For any integer $p'\in(p,p+2^{u+1})$, there exists an integer $v\in[0,u]$ such that $p = 2^v  \floor{\frac{p'}{2^v }}-2^v$.
\item For any integer $p'\geq p+2^{u+1}$, there exists no integer $v\geq0$ such that $p = 2^v  \floor{\frac{p'}{2^v }}-2^v$.
\end{enumerate}
In other words, to remove meeting point $p=(2w+1)2^u$, a party should reach $\Point=p+2^{u+1}$.
\end{lemma}

\begin{proof}
Let $\Delta p=p'-p$.
For Case 1,
considering $v\in[0,u]$,
we have
$2^v  \floor{\frac{p'}{2^v }}-2^v
=p + 2^v \floor{\frac{\Delta p}{2^v }} -2^v$.
Then, since $1\leq\Delta p\leq 2^{u+1}-1$,
there always exists an integer $v\in[0,u]$
such that $\floor{\frac{\Delta p}{2^v }}=1$
and thus $2^v  \floor{\frac{p'}{2^v }}-2^v =p$.

For Case 2,
if $v\in[0,u]$,
then we have
$2^v  \floor{\frac{p'}{2^v }}-2^v
=p + 2^v \floor{\frac{\Delta p}{2^v }} -2^v
\geq p + 2^v \floor{\frac{2^{u+1}}{2^v }} -2^v
=p + 2^{u+1} -2^v>p
$.
Thus, 
we should consider $v\geq u+1$.
Suppose there exists an integer
$v\geq u+1$ such that $p = 2^v  \floor{\frac{p'}{2^v }}-2^v$.
Then, we should have $p=w'2^v=(2w+1)2^u$
for a positive integer $w'$,
which is a contradiction
since $2w+1$ is not even.
Therefore, there exists no integer $v\geq0$ such that $p = 2^v  \floor{\frac{p'}{2^v }}-2^v$.
\end{proof}

The following lemma gives an explicit analysis
of what happens when some meeting point is unavailable.

\begin{lemma}
[Implication of Missing True Meeting Points]
\label{lem:valuesForMissingMP2}
Suppose $\verif=2^u$ for some non-negative integer $u$, and
$\forall i\in AB: \verif_i = \verif> \disagree_{AB}$.
Moreover, suppose any true meeting point
that is a multiple of $\verif$ (as promised in Lemma~\ref{lem:haveATrueMP})
is not saved by at least one party, i.e.,
for that party, $\HD[\MPtwo] = \bot$.

Then, we have $\alpha_{AB} \geq 0.2\verif$ or $\gamma\geq0.5\verif$.
\end{lemma}

\begin{figure}[h]
\centering
\includegraphics[width=3.5in]{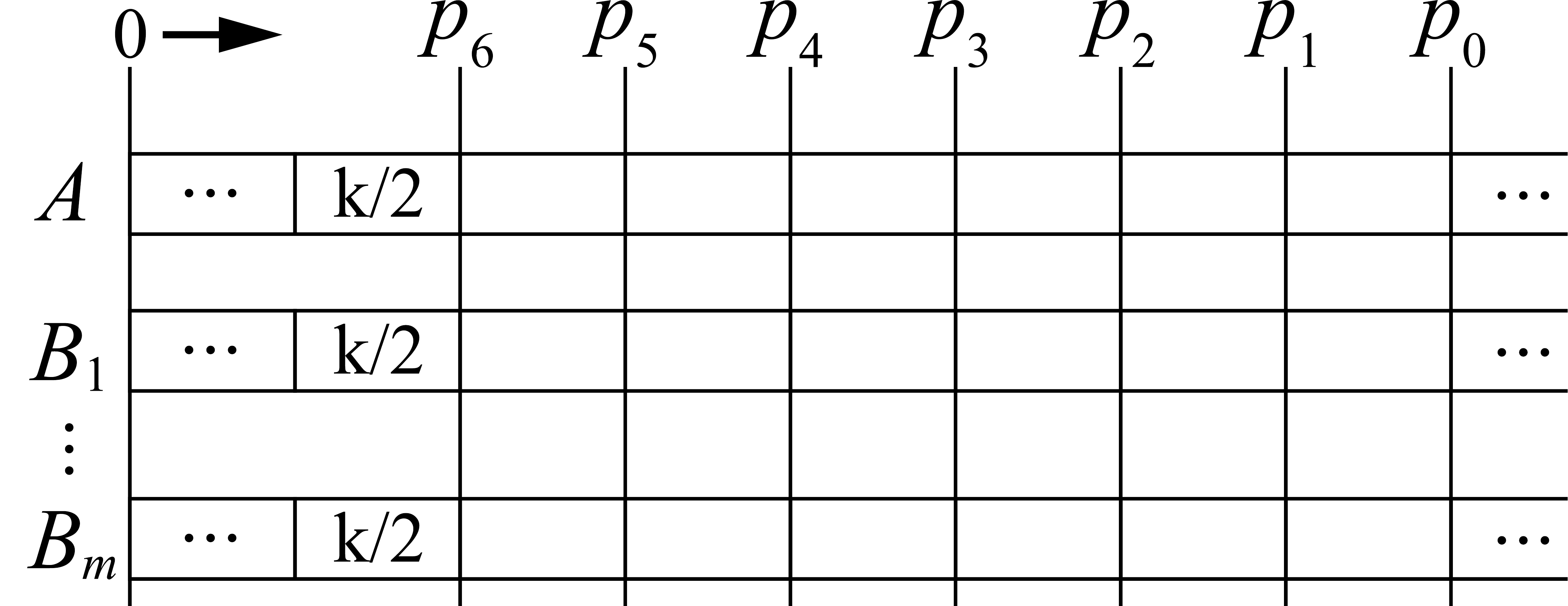}
\caption{Position illustration.}
\label{fig:missingMP}
\end{figure}

\begin{proof}
Recall that  $\Point_i$ is the current position of party $i\in AB$.
Let $w$ be a non-negative integer such that $w\verif \leq\agree <(w+1)\verif$.
According to the proof of Lemma \ref{lem:haveATrueMP},
let $\sset_2$ denote the set of parties $i$
such that $(w+1)\verif \leq \Point_i <(w+2)\verif$.
Then, there exists at least one party in $\sset_2$
whose $\MPtwo$ at scale~$\verif$ (which is $w \verif$) is missing.
Since $w \verif$ is missing, by Fact~\ref{fact:MPone},
$w \verif$ cannot be meeting point one of any other scale,
which implies that $w$ is odd; this allows us to use Lemma~\ref{lem:missMP} later.

Let $\sset_m\subseteq \sset_2$
denote the set of parties whose $\MPtwo$ is missing.
In Figure~\ref{fig:missingMP},
let $p_4$ be $w\verif$ and each block has length $\frac{\verif}{2}$.
Then, we have $\agree \in[p_4,p_2)$.
For each party $j\in \sset_m$,
we have $\Point_j\in[p_2,p_0)$, $\MPone=p_2$ and $\MPtwo=p_4$ is missing.

Let $t$ be the epoch when $\forall i\in AB: \verif_i = \verif> \disagree_{AB}$ but there is no true meeting point
that is a multiple of $\verif$ and is saved by all parties.
By Lemma \ref{lem:missMP},
each party $j\in \sset_m$ should reach at least $p_0$ before epoch $t$ to remove $p_4$.
For each party $j\in \sset_m$, let $t_j$ be the last epoch when $j$ rolls back from a position at least $p_0$ to $[p_2,p_0)$ before epoch $t$.
Suppose party $x$ is one of the parties with minimum $t_x$ in $\sset_m$.
Let $t_D$ be the last epoch when $D_{AB}=0$ before epoch $t_x$
and $p_D$ be the $\agree$ value at epoch $t_D$.

\begin{claim}
\label{clm:onlyInP2P0}
For each party $j\in \sset_m$, the party $j$ has no meeting point in $[p_4,p_2)$ during the epoch period $[t_j,t]$.
Moreover, the party $j$ can only be in $[p_2,p_0)$ during the epoch period $(t_j,t]$.
\end{claim}
\begin{proof}
For each party $j\in \sset_m$, by Lemma \ref{lem:missMP},
$j$ has at most one meeting point $p_6 := (w-1) \verif$
in $[p_6,p_2)$ to roll back right after epoch $t_j$.
Thus, $j$ can only roll back to a position at most $p_6$ and go forward to recover meeting points in $[p_4,p_2)$ .
But then $j$ also recovers $p_4$ and must reach at least $p_0$ again to remove $p_4$, which violates the definition of $t_j$.
Therefore, $j$ cannot roll back to less than $p_2$
and can only be in $[p_2,p_0)$.
\end{proof}

Then, we consider the following cases:

\begin{enumerate}
\item
If $p_{D}\geq p_0$, then $\agree$ cannot be in $[p_4,p_2)$ during the epoch period $(t_D,t)$.
We prove this result by contradiction.
First, the party $x$ can only be in a position at least $p_0$ during the epoch period $[t_D,t_x]$ since $p_{D}\geq p_0$.
Then, by Claim \ref{clm:onlyInP2P0}, $x$ can only be in $[p_2,p_0)$ after $t_x$.
Thus, $x$ cannot decrease $\agree$ to be in $[p_4,p_2)$.
Suppose party $j\not=x$ is the first one that rolls back and decreases $\agree$ to be in $[p_4,p_2)$  after $t_D$.
Then, by Lemma \ref{lem:missMP}
and $p_D\geq p_0$, at first $j$ only has at most one meeting point $p_6$
in $[p_6,p_2)$ to roll back after $t_D$.
Thus, $j$ can only first roll back to a position at most $p_6$
and thus set $\agree \leq p_6$.
Since $\agree$ should finally be in $[p_4,p_2)$,
$\agree$ should grow into $[p_4,p_2)$ after $j$ rolls back.
However, 
$x$ can only be in $[p_2,p_0)$ during the epoch period $[t_x,t]$.
Without $x$, $\agree$ cannot grow up to $[p_4,p_2)$ after $j$ rolls back,
which is a contradiction.
Therefore, we must have $p_D<p_0$.

\item
If $p_{D}\in [p_1,p_0)$ and $\agree\geq p_2$ at epoch $t_x$,
then there exists a party $j\not=x$
such that $j$ rolls back and sets $\agree \in [p_4,p_2)$
during the time period $(t_x,t]$.
By Lemma \ref{lem:missMP}
and $p_D\in [p_1,p_0)$, at first $j$ only has one meeting point $p_4$
in $[p_4,p_2)$ to roll back after $t_D$.
Thus, $j$ can only first roll back to a position at most $p_4$
before setting $\agree \in [p_4,p_2)$.
Since by Claim \ref{clm:onlyInP2P0} $x$ can only be in $[p_2,p_0)$ after $t_x$,
$j$ can only roll back to a position at most $p_4$.

Observe that for party~$x$, if there is an available meeting point at most $p_4$, then it is at most $p_6$.
If party~$j$ rolls back to $p_6$, this will cause $\agree$ to drop to at most $p_6$, 
which causes a contradiction as in the previous case.  
Hence, party~$j$ has to roll back to $p_4$, which is missing for party~$x$.
Therefore, there must be at least
$0.4\verif$ bad votes,
i.e. $\beta_j\geq 0.4\verif$.
Then, according to how $\alpha$ accumulates bad votes from $\beta$, we have $\alpha_{AB} \geq 0.2\verif$ at epoch $t$.
Notice that after $j$ rolls back, $\disagree_{AB} \not= 0$ until $t$.

\item If $p_{D}\in [p_1,p_0)$ and $\agree < p_2$ at epoch $t_x$,
then there exists a party $j\in AB$
such that $j$ rolls back and sets $\agree < p_2$
during the epoch period $(t_D,t_x]$.
By Lemma \ref{lem:missMP}
and $p_D\in [p_1,p_0)$, at first $j$ only has one meeting point $p_4$
in $[p_4,p_2)$ to roll back after $t_D$.
Thus, we conclude that $j$ can only first roll back to a position at most $p_4$
before setting $\agree < p_2$;
as in the previous case, actually $j$ has to roll back to $p_4$.  

Since $\disagree_{AB} \not= 0$ during $(t_D,t_x]$ according to the definition of $t_D$,
$j$ can only roll back because of at least
$0.4\verif$ bad votes,
i.e. $\beta_j\geq 0.4\verif$.
Then, according to how $\alpha$ accumulates bad votes from $\beta$, we have $\alpha_{AB} \geq 0.2\verif$ at epoch $t$.
Notice that after $j$ rolls back, $\disagree_{AB} \not= 0$ until $t$.

\item If $p_{D}< p_1$, then the party $x$ should suffer at least
$\frac{\verif}{2}$ corrupted computations to reach at least $p_0$ during the epoch period $(t_D,t_x]$, i.e. $\gamma\geq 0.5\verif$.

Remark that this portion of $\gamma$ will not decrease until $t$.
Recall that the condition for $\gamma$ to decrease
is as follows:
\emph{
All parties do meeting point transitions with $\forall i\in AB: 0<2\disagree_{AB}<\verif_{i}=\verif'$ and $\alpha_{AB}+\beta_{AB} < 0.1\verif'$ before the transition and with $\disagree_{AB}=0$ after the transition.  Then, $\gamma$ decreases by $0.25\verif'$.
}

According to the definition of $t_D$,
$\disagree_{AB}=0$ can only happen again during $(t_x,t)$.
Then, the epoch satisfying the above condition can only happen during $(t_x,t)$.
Since $\forall i\in AB: 2\disagree_{AB}<\verif_{i}=\verif'$, by Lemma \ref{lem:haveATrueMP},
there should be another epoch $t'\in(t_x,t)$ with $\forall i\in AB: \verif_i=\frac{\verif'}{2} > \disagree_{AB}$ such that
all parties should have a true meeting point.
Moreover, that true meeting point is not saved by at least one party.
Otherwise, the parties should suffer more than $0.1\cdot\frac{\verif}{2}$ bad votes to that true meeting point, i.e. $\forall i\in AB: \beta_i>0.05\verif'$ and $\gamma$ does not decrease.
Then, we can introduce another instance of the above case analysis.
Since one of the condition is $\alpha_{AB}+\beta_{AB} < 0.1\verif'$,
we can directly go to case 4 showing that there should be some party $x'$
(with similar definition to $x$) suffering at least $0.5\cdot\frac{\verif'}{2}=0.25\verif'$ corrupted computations during $(t_{D'},t_{x'}]$ (with similar definition to $t_{D},t_{x}$).
Observe that $(t_D,t_x]\cap(t_{D'},t_{x'}]=\emptyset$ according the definition of $t_D$ and $t_{D'}$.
Thus, the decrease $0.25\verif'$ of $\gamma$ will not cost the portion $0.5\verif$ obtained during $(t_D,t_x]$.

\end{enumerate}

Therefore, we conclude that 
$\alpha_{AB} \geq 0.2\verif$ or $\gamma\geq0.5\verif$.
\end{proof}

We are ready to prove the first statement of Lemma~\ref{lemma:potential_tech}.

\begin{lemma}[Lower Bound on Potential Function:
Similar to Lemma 7.4 in~\cite{Haeupler14}]
\label{lem:potentialIncreaseEpoch}
In every epoch,
if there exists at least one corruption or short hash collision,
the potential decreases at most by $(m+1)(2C_6-C_2) = O(m^4)$.
Furthermore, if the corresponding epoch is consistent and no corruption occurs,
the potential increases by at least one.
If the corresponding epoch is inconsistent and no corruption or short hash collision occurs,
the potential increases by at least two.
\end{lemma}

\begin{proof}
We denote with $\agree $, $\disagree $, $\verif$, $\error$, $\alpha$, $\beta$ and $\gamma$ the values right before the transition stage 
and denote with $\agree' $, $\disagree' $, $\verif'$, $\error'$, $\alpha'$, $\beta'$ and $\gamma'$ the values after the transition stage.
We also use a $\Delta$ in front of any variable to denote the change of value to this variable during the transition stage. 

Given Lemma \ref{lem:potentialIncreaseVCPhase}, it suffices to show that a transition stage never decreases the potential, i.e. $\Delta\Phi\geq0$ . We show exactly this, except for one case, in which the potential decreases by a small constant. In case of the epoch being error and short hash collision free, this constant is shown to be less than the increase of the preceding computation and verification stage and the total potential increase for that (inconsistent) epoch is at least two.

We now make the following case distinction according to which combination of transition(s) occurred in the epoch and whether or not the parties agreed in their $\verif$ parameter before the transition:

\begin{enumerate}
\item All parties do normal transitions.

Then, no matter $\forall i,j\in AB:\verif_i = \verif_{j}$ or not, all variables measured in the potential do not change. Thus, $\Delta\Phi=0$.

\item $\exists i,j\in AB:\verif_i \not= \verif_{j}$ and at least one party does non-normal transition.

Then, we use \eqref{eq:potential_neq} to evaluate the potential before the transition stage.
In addition, the $\gamma$ does not change,
so we can ignore the $\gamma$ part in the potential.
From \eqref{eq:potential_neq}, the potential before the transition stage is
$$
\Phi=\agree - C_3 \disagree_{AB} - 0.9 C_4  \sum_{i\in AB}\verif_{i}  +  C_4  \sum_{i\in AB}\error_{i} -  C_6  \sum_{i\in AB}\alpha_{i} -  C_6  \sum_{i\in AB}\beta_{i}.
$$

Let $\sset\subseteq AB$ be the set of parties that do non-normal transitions.
If $\sset \not= AB$, then at least one party does normal transition (and 
in this case, also at least one party does non-normal transition).
Hence,
 $\exists i,j\in AB:\verif_i' \not= \verif_{j}'$ after the transition stage and from \eqref{eq:potential_neq} the potential is
$$
\Phi'=\agree' - C_3 \disagree'_{AB} - 0.9 C_4 \sum_{i\in AB \setminus \sset} \verif_{i}  +  C_4 \sum_{i\in AB \setminus \sset} \error_{i} -  C_6  \sum_{i\in AB \setminus \sset}\alpha_{i} -  C_6  \sum_{i\in \sset}\alpha_{i}' -  C_6  \sum_{i\in AB \setminus \sset}\beta_{i}.
$$
If $\sset = AB$, then $\forall i\in AB:\verif_i' = \error_{i}' = \beta_{i}' =0$ after the transition stage and from \eqref{eq:potential_eq} the potential is
$$
\Phi'=\agree' - C_3 \disagree'_{AB} -  C_6  \sum_{i\in \sset}\alpha_{i}'.
$$

Observe that in the first equation for $\Phi'$,
when $S = AB$, it reduces to the second equation for $\Phi'$;
hence, we can just use the first equation for $\Phi'$.

We consider the following cases.

\begin{enumerate}

\item Suppose in this epoch, there exists at least one corruption or short hash collision.

The change of the potential can be expressed:
$$
\Delta\Phi=\Delta\agree - C_3 \Delta\disagree_{AB} + 0.9 C_4 \sum_{i\in \sset} \verif_{i}  -  C_4 \sum_{i\in \sset} \error_{i} -  C_6  \sum_{i\in  \sset}\Delta\alpha_{i} +  C_6  \sum_{i\in \sset}\beta_{i}.
$$
According to how $\alpha$ accumulates bad votes from $\beta$, we have $\forall i\in \sset: \Delta\alpha_{i}=\alpha_{i}'-\alpha_{i}\leq 0.5\beta_{i}$.
Then, we have $-  C_6  \sum_{i\in  \sset}\Delta\alpha_{i} +  C_6  \sum_{i\in \sset}\beta_{i} \geq - 0.5 C_6 \sum_{i\in  \sset} \beta_{i} +  C_6  \sum_{i\in \sset}\beta_{i} \geq 0$.
By Lemma \ref{lem:ELessThanK}, we also have $\error_i \leq 0.5(\verif_i + 1)$ for each $i\in \sset$.
Moreover, by Lemma \ref{lem:changeOfAgree}, the decrease of $\agree$ is at most $2\max_{i\in \sset}\verif_i-1$
and the increase of $\disagree_{AB}$ is at most $m(2\max_{i\in \sset}\verif_i-1)$.
Thus, we have

\begin{align*} 
\Delta\Phi & \geq  -(2\max_{i\in \sset}\verif_i-1) - mC_3 (2\max_{i\in \sset}\verif_i-1) + 0.9 C_4 \sum_{i\in \sset} \verif_{i} - 0.5 C_4 \sum_{i\in \sset} (\verif_i + 1) \\
& =  - (2m C_3 + 2) \max_{i\in \sset}\verif_i + (m C_3 + 1) + 0.4 C_4 \sum_{i \in S} \verif_i - 0.5 C_4 \cdot |S|.
\end{align*}

Now, we further consider the following different scenarios:
\begin{enumerate}
\item $\max_{i\in \sset}\verif_i \geq 2$.

Let $\sset_2$, $\sset_1$ and $\sset_0$ denote the set of parties in $S$ with $\verif\geq2$, $\verif=1$ and $\verif=0$ respectively.
Then, with $C_4 \geq 5(1+mC_3)$, the coefficient of $\max_{i\in \sset}\verif_i$ is non-negative.
Hence, we can obtain a lower bound of $\Delta\Phi$ by replacing $\max_{i\in \sset}\verif_i$ with 2,
and truncate values of $\verif_i$ to 2:
$$
\Delta\Phi\geq 0.3 C_4|S_2| -3 - 3mC_3 - 0.1 C_4 |S_1| - 0.5 C_4 |S_0|.
$$

According the Case \ref{case:neq2neq} in the proof of Lemma \ref{lem:potentialIncreaseVCPhase},
the potential can only decrease at most by
\begin{align*}
&3+3mC_3+(m+1)C_6+C_7+(0.5C_4+C_3)|S_0|+C_4|S_1|+0.6C_4|S_2|\\
\leq&3+3mC_3+(m+1)C_6+C_7+(m+1)C_4\leq(m+1)(2C_6-C_2)
\end{align*}
in the corresponding epoch.

\item
$\max_{i\in \sset}\verif_i \leq 1$.

Then, we can obtain a lower bound of $\Delta\Phi$ by replacing $\max_{i\in \sset}\verif_i$ with 1:
$$
\Delta\Phi\geq -1 - mC_3 - 0.1 C_4 |S_1| - 0.5 C_4 |S_0|
$$
where $\sset_1$ and $\sset_0$ denote the set of parties in $S$ with $\verif=1$ and $\verif=0$ respectively.
According the Case \ref{case:eq2neq} in the proof of Lemma \ref{lem:potentialIncreaseVCPhase},
the potential can only decrease at most by
\begin{align*}
&1+mC_3+(m+1)C_6+C_7+(0.5C_4+C_3)|S_0|+C_4|S_1|\\
\leq&1+mC_3+(m+1)C_6+C_7+(m+1)C_4\leq(m+1)(2C_6-C_2)
\end{align*}
in the corresponding epoch.

\ignore{
\item
If no error or short hash collision occurs and $\min_{i\in \sset}\verif_i \geq 2$, then we have
$$
\Delta\Phi \geq - 3 - 3mC_3 + 0.3C_4|\sset| \geq 0
$$
with $C_4 \geq 10(1 + mC_3)/|\sset|$.
}

\end{enumerate}

\item
Suppose no corruption or short hash collision occurs.

Then, $\Delta\Phi$ can be negative.
We need to include the increase of the potential during the preceding
computation and verification stages to make sure the total increase is at least 2, since this epoch is inconsistent.
We use the notations with a superscript $*$ to denote the values at the beginning of an epoch.
Since no corruption or short hash collision happens and $\exists i,j\in AB:\verif_i \not= \verif_{j}$,
all parties go to the branch that increases $\error$ and do one dummy epoch simulation.

This also implies $\min_{i\in AB}\verif_i \geq 1$.
Then, we have $\forall i\in AB: \verif^*_i = \verif_i-1$
and $\error^*_i = \error_i-1$ while the other parameters do not change.
Thus, from \eqref{eq:potential_neq} we have
$$
\Phi^*=\agree - C_3 \disagree_{AB} - 0.9 C_4  \sum_{i\in AB}(\verif_{i}-1)  +  C_4  \sum_{i\in AB}(\error_{i}-1) -  C_6  \sum_{i\in AB}\alpha_{i} -  C_6  \sum_{i\in AB}\beta_{i}.
$$
Then, the change of the potential in this epoch is
$$
\Phi' - \Phi^* = \Delta\agree - C_3 \Delta\disagree_{AB}
+ 0.9 C_4  \sum_{i\in S}\verif_{i} - C_4  \sum_{i\in S}\error_{i}
+ 0.1 C_4 |AB| -  C_6  \sum_{i\in  \sset}\Delta\alpha_{i} +  C_6  \sum_{i\in \sset}\beta_{i}.
$$
According to how $\alpha$ accumulates bad votes from $\beta$, we have $\forall i\in \sset: \Delta\alpha_{i}=\alpha_{i}'-\alpha_{i}\leq 0.5\beta_{i}$.
Then, we have
$$-  C_6  \sum_{i\in  \sset}\Delta\alpha_{i} +  C_6  \sum_{i\in \sset}\beta_{i} \geq - 0.5 C_6 \sum_{i\in  \sset} \beta_{i} +  C_6  \sum_{i\in \sset}\beta_{i} \geq 0.
$$
By Lemma \ref{lem:ELessThanK}, we also have $\error_i \leq 0.5(\verif_i + 1)$ for each $i\in \sset$.
Moreover, by Lemma \ref{lem:changeOfAgree}, the decrease of $\agree$ is at most $2\max_{i\in \sset}\verif_i-1$
and the increase of $\disagree_{AB}$ is at most $m(2\max_{i\in \sset}\verif_i-1)$.
Thus, we have 
$$
\Phi' - \Phi^* \geq -(2\max_{i\in \sset}\verif_i-1)
- mC_3(2\max_{i\in \sset}\verif_i-1)
+0.4 C_4 \sum_{i\in \sset}\verif_i
-0.4 C_4 |\sset|
+0.1 C_4 |AB\setminus\sset|.
$$

If $S=AB$,
then we have $\max_{i\in \sset}\verif_i\geq2$ since $\exists i,j\in AB:\verif_i \not= \verif_{j}$.
With $C_4 \geq 5(1+mC_3)$,
we can obtain a lower bound of $\Phi' - \Phi^*$ by replacing $\max_{i\in \sset}\verif_i$ with 2.
Then, we have 
$$
\Phi' - \Phi^* \geq -3
-3 mC_3
+0.4 C_4
\geq 2,
$$
with $C_4\geq5(5+3mC_3)/2$.

Otherwise, if $S\not=AB$,
then we have $\max_{i\in \sset}\verif_i\geq1$.
With $C_4 \geq 5(1+mC_3)$,
we can obtain a lower bound of $\Phi' - \Phi^*$ by replacing $\max_{i\in \sset}\verif_i$ with 1.
Then, we have 
$$
\Phi' - \Phi^* \geq -1
- mC_3
+0.1 C_4
\geq 2,
$$
with $C_4\geq10(3+mC_3)$.
\end{enumerate}

\item $\forall i,j\in AB: \verif_i = \verif_j = \verif$ and at least one party does error transition.

Then,
we can ignore the $\gamma$ part in the potential
since $\gamma$ does not change in this case.
From \eqref{eq:potential_eq} the potential before the transition stage is
$$
\Phi=\agree - C_3 \disagree_{AB} + C_2 \sum_{i\in AB} \verif_{i}  -  C_5  \sum_{i\in AB} \error_{i}
-C_6 \sum_{i\in AB} \alpha_{i} - 2C_6 \sum_{i\in AB} \beta_{i}.
$$
Let $\sset_e,\sset_m,\sset_n \subseteq AB$ be the set of parties that do error transitions, meeting point transitions and normal transitions respectively.
Then, we consider the following two scenarios:

\begin{enumerate}
\item
If $\sset_n\not=\emptyset$,
then we should use \eqref{eq:potential_neq} for the potential.
Hence,  the potential after the transition stage is
$$
\Phi'=\agree' - C_3 \disagree_{AB}' - 0.9 C_4
\sum_{i\in \sset_n} \verif_i + C_4 \sum_{i\in \sset_n} \error_i - C_6 \sum_{i\in \sset_m} \alpha_{i}' - C_6 \sum_{i\in AB\setminus\sset_m} \alpha_{i} - C_6\sum_{i\in \sset_n} \beta_{i}.
$$
The change of the potential consists of two parts: 
\begin{align*}
\Delta\Phi=&\Delta\agree - C_3\Delta\disagree_{AB} - 0.9 C_4
\sum_{i\in \sset_n} \verif_i - C_2 \sum_{i\in AB} \verif_{i}
+ C_4 \sum_{i\in \sset_n} \error_i  +C_5  \sum_{i\in AB} \error_{i}\\
&- C_6 \sum_{i\in \sset_m} \Delta\alpha_{i} + C_6\sum_{i\in \sset_n} \beta_{i} + 2C_6\sum_{i\in AB\setminus\sset_n} \beta_{i}.
\end{align*}

We analyze the second part first.
According to how $\alpha$ accumulates bad votes from $\beta$, we have $\forall i\in \sset_m: \Delta\alpha_{i}=\alpha_{i}'-\alpha_{i}\leq 0.5\beta_{i}$.
Then, we have
$$
- C_6 \sum_{i\in \sset_m} \Delta\alpha_{i} + C_6\sum_{i\in \sset_n} \beta_{i} + 2C_6\sum_{i\in AB\setminus\sset_n} \beta_{i} \geq 1.5 C_6 \sum_{i\in \sset_m} \beta_{i} + C_6\sum_{i\in \sset_n} \beta_{i} + 2C_6\sum_{i\in \sset_e} \beta_{i} \geq 0.$$

Hence, it remains to analyze the first part.
By the condition of the error transition, we also have $2\error_i \geq \verif$ for each $i\in \sset_e$.
Moreover, by Lemma \ref{lem:changeOfAgree}, the decrease of $\agree$ is at most $2\verif-1$
and the increase of $\disagree_{AB}$ is at most $m(2\verif-1)$.
Thus, we have
\begin{align*}
\Delta\Phi\geq& -(2\verif-1) - m C_3 \cdot (2\verif-1) - 0.9 C_4 \cdot
|\sset_n| \verif - (m+1) C_2 \verif\\
&+ C_4 \sum_{i\in \sset_n} \error_i  + C_5  \sum_{i\in AB\setminus\sset_e} \error_{i}
+ 0.5 C_5  \cdot |\sset_e| \cdot \verif\\
=& 1 + mC_3 + C_4 \sum_{i\in \sset_n} \error_i + C_5  \sum_{i\in AB\setminus\sset_e} \error_{i}\\
&+[0.5C_5 \cdot |\sset_e| -2-2mC_3-0.9C_4 \cdot |\sset_n| -(m+1)C_2]\verif
\geq 0
\end{align*}
with $C_5 \geq [4 + 4mC_3+1.8C_4|\sset_n| +2(m+1)C_2]/|\sset_e|$.

\item
If $\sset_n = \emptyset$,
then we should use \eqref{eq:potential_eq} for the potential.
Hence,
the potential after the transition stage is
$$
\Phi'=\agree' - C_3 \disagree_{AB}' - C_6 \sum_{i\in \sset_m} \alpha_{i}' - C_6 \sum_{i\in \sset_e} \alpha_{i}.
$$
Then, the change of the potential is 
\begin{align*}
\Delta\Phi=&\Delta\agree - C_3 \Delta \disagree_{AB}- C_2 \sum_{i\in AB} \verif_{i}  +  C_5  \sum_{i\in AB} \error_{i}
-C_6 \sum_{i\in \sset_m} \Delta\alpha_{i} + 2C_6 \sum_{i\in AB} \beta_{i}.
\end{align*}
According to how $\alpha$ accumulates bad votes from $\beta$, we have $\forall i\in \sset_m: \Delta\alpha_{i}=\alpha_{i}'-\alpha_{i}\leq 0.5\beta_{i}$.
Then, we have
$$
- C_6 \sum_{i\in \sset_m} \Delta\alpha_{i} + 2C_6\sum_{i\in AB} \beta_{i}
\geq 1.5 C_6 \sum_{i\in \sset_m} \beta_{i} + 2C_6\sum_{i\in \sset_e} \beta_{i} \geq 0.$$
By the condition of the error transition, we also have $2\error_i \geq \verif$ for each $i\in \sset_e$.
Moreover, by Lemma \ref{lem:changeOfAgree}, the decrease of $\agree$ is at most $2\verif-1$
and the increase of $\disagree_{AB}$ is at most $m(2\verif-1)$.
Thus, we have
\begin{align*}
\Delta\Phi\geq& -(2\verif-1) - mC_3(2\verif-1) - (m+1) C_2 \verif  +  0.5 C_5  |\sset_e| \verif
+  C_5  \sum_{i\in \sset_m} \error_{i}\\
\geq& 1 + mC_3 +  C_5  \sum_{i\in \sset_m} \error_{i}
+[0.5 C_5  |\sset_e| -2 -2mC_3 - (m+1) C_2]\verif
\geq 0
\end{align*}
with $C_5 \geq [4 + 4mC_3 +2(m+1)C_2]/|\sset_e|$.
\end{enumerate}

\item $\forall i,j\in AB: \verif_i = \verif_j = \verif$, and no error transition occurs, and at least one party does meeting point transition but not all.

Then,
we can ignore the $\gamma$ part in the potential since $\gamma$ does not change.
From \eqref{eq:potential_eq} the potential before the transition stage is
$$
\Phi=\agree - C_3 \disagree_{AB} + C_2 \sum_{i\in AB} \verif_{i}  -  C_5  \sum_{i\in AB} \error_{i}
-C_6 \sum_{i\in AB} \alpha_{i} - 2C_6 \sum_{i\in AB} \beta_{i}.
$$
Let $\sset_m,\sset_n \subseteq AB$ be the set of parties that do meeting point transitions and normal transitions respectively.
Then,  after the transition stage,
we should use \eqref{eq:potential_neq} for the potential:

$$
\Phi'=\agree' - C_3 \disagree_{AB}' - 0.9 C_4
\sum_{i\in \sset_n} \verif_i + C_4 \sum_{i\in \sset_n} \error_i - C_6 \sum_{i\in \sset_m} \alpha_{i}' - C_6 \sum_{i\in \sset_n} \alpha_{i} - C_6\sum_{i\in \sset_n} \beta_{i}.
$$
Then, the change of the potential consists of two parts: 
\begin{align*}
\Delta\Phi=&\Delta\agree - C_3\Delta\disagree_{AB} - 0.9 C_4
\sum_{i\in \sset_n} \verif_i - C_2 \sum_{i\in AB} \verif_{i}
+ C_4 \sum_{i\in \sset_n} \error_i  +C_5  \sum_{i\in AB} \error_{i}\\
&- C_6 \sum_{i\in \sset_m} \Delta\alpha_{i} + C_6\sum_{i\in \sset_n} \beta_{i} + 2C_6\sum_{i\in \sset_m} \beta_{i}.
\end{align*}

We first show that the second part is non-negative.
According to how $\alpha$ accumulates bad votes from $\beta$, we have $\forall i\in \sset_m: \Delta\alpha_{i}=\alpha_{i}'-\alpha_{i}\leq 0.5\beta_{i}$.
Then, we have
$$
- C_6 \sum_{i\in \sset_m} \Delta\alpha_{i} + C_6\sum_{i\in \sset_n} \beta_{i} + 2C_6\sum_{i\in \sset_m} \beta_{i}
\geq  C_6\sum_{i\in \sset_n} \beta_{i} + 1.5C_6\sum_{i\in \sset_m} \beta_{i} \geq 0.
$$

Hence, it suffices to analyze the first part of $\Delta \Phi$.
By Lemma \ref{lem:changeOfAgree}, the decrease of $\agree$ is at most $2\verif-1$
and the increase of $\disagree_{AB}$ is at most $m(2\verif-1)$.
Thus, we have
\begin{align*}
\Delta\Phi\geq& -(2\verif-1) - mC_3(2\verif-1) - 0.9 C_4
|\sset_n| \verif - (m+1) C_2 \verif\\
&+ C_4 \sum_{i\in \sset_n} \error_i  + C_5  \sum_{i\in AB} \error_{i} + C_6\sum_{i\in \sset_n} \beta_{i} + 1.5C_6\sum_{i\in \sset_m} \beta_{i}\\
=& 1 + mC_3 + C_4 \sum_{i\in \sset_n} \error_i + C_5  \sum_{i\in AB} \error_{i} + C_6\sum_{i\in \sset_n} \beta_{i} + 1.5C_6\sum_{i\in \sset_m} \beta_{i}\\
&+[-2-2mC_3-0.9C_4|\sset_n| -(m+1)C_2]\verif.
\end{align*}

\begin{enumerate}
\item
There exists a true meeting point that every party has saved.

Then, each party in $\sset_n$ fails to transition to that meeting point, because of totally more than $0.1\verif$ bad votes.
Recalling that all $\beta_i$'s increase together, we have $\forall i\in AB: \beta_i > 0.1\verif$.

\ignore{
\item
If there exists a true meeting point and Alice does normal transition,
then Alice should fail to transition to that true meeting point because of more than $0.1\verif$ bad votes to Alice, i.e. $\forall i\in AB: \beta_i > 0.1\verif$.

\item
If there exists a true meeting point and Alice transitions to a false meeting point,
then Alice should suffer at least $0.4\verif$ bad votes to that false meeting point, i.e. i.e. $\forall i\in AB: \beta_i\geq0.4\verif$.
}

\item
There is no true meeting point that everyone has saved.

Then, each party in $\sset_m$ should suffer at least $0.4\verif$ bad votes to that meeting point, i.e. $\forall i\in AB: \beta_i\geq0.4\verif$.
\end{enumerate}
Therefore, we have $\forall i\in AB: \beta_i > 0.1\verif$.
Then, we have $\Delta\Phi\geq 0$
with $C_6 \geq [40+40mC_3+18C_4|\sset_n| +20(m+1)C_2]/(2|\sset_n|+3|\sset_m|)$.

\item The remaining case is that $\forall i,j\in AB: \verif_i = \verif_j = \verif$ and all parties do meeting point transitions.

Before the transition stage,
we should use \eqref{eq:potential_eq} for the potential:
$$
\Phi=\agree - C_3 \disagree_{AB} + C_2  \verif_{AB}  -  C_5  \error_{AB} -C_6\alpha_{AB}-2C_6\beta_{AB} - C_7\gamma.
$$

We further consider the following cases.

\begin{enumerate}
\item $\disagree'_{AB}  \neq 0$.

This means Alice and some Bob transition to different meeting points.

From Alice's perspective, there are two candidates $\MPone$ and $\MPtwo$.
Observe that among $0.5 \verif$ votes, Alice must give $0.4 \verif$ votes for a candidate
point, in order for that transition to happen.  Since all parties prefer $\MPone$ if both are available,
it follows that if Alice votes for one candidate in an epoch with no corruption or hash collision,
then all other parties should vote for the same candidate.

Hence, it follows there must be at least $0.4 \verif$ bad votes,
if Alice and some Bob transition to different meeting points.
Thus, we have $\forall i\in AB: \beta_i\geq0.4\verif$ and the $\gamma$ does not change.
Then, from \eqref{eq:potential_eq} the potential after the transition stage is
$$
\Phi'=\agree' - C_3 \disagree'_{AB} -C_6\alpha'_{AB}- C_7\gamma.
$$
According to how $\alpha$ accumulates bad votes from $\beta$,
we have $\alpha'_{AB}-\alpha_{AB}= 0.5\beta_{AB}$.
By Lemma \ref{lem:changeOfAgree}, the decrease of $\agree$ is at most $2\verif-1$
and the increase of $\disagree_{AB}$ is at most $m(2\verif-1)$.
Then, we have
\begin{align*}
\Delta \Phi &=
\Delta\agree-C_3\Delta\disagree_{AB}-C_2\verif_{AB}
+C_5\error_{AB}-C_6(\alpha'_{AB}-\alpha_{AB})+2C_6\beta_{AB}
\\
&\geq -(2\verif-1)-mC_3(2\verif-1)-(m+1)C_2\verif
+C_5\error_{AB}-0.5C_6\beta_{AB}+2C_6\beta_{AB}
\\
&\geq1+C_3+C_5\error_{AB}+[0.6(m+1)C_6-2-2mC_3-(m+1)C_2]\verif
\geq0,
\end{align*}
with $C_6\geq [10+5(m+1)C_2+10mC_3]/3(m+1)$.

\item
$\disagree'_{AB}  = 0$ and $\verif \leq 2 \disagree_{AB}$.

Then,
the $\gamma$ does not change.
From \eqref{eq:potential_eq} the potential after the transition stage is
$$
\Phi'=\agree' - C_7\gamma.
$$
By Lemma \ref{lem:changeOfAgree}, the decrease of $\agree$ is at most $2\verif-1$.
Then, we have
\begin{align*}
\Delta \Phi =&\Delta\agree+C_3\disagree_{AB}-C_2\verif_{AB}+C_5\error_{AB}
+C_6\alpha_{AB}+2C_6\beta_{AB}
\\
\geq&-(2\verif-1)
+0.5C_3\verif -(m+1)C_2\verif+C_5\error_{AB}
+C_6\alpha_{AB}+2C_6\beta_{AB}
\\
=& 1+C_5\error_{AB}
+C_6\alpha_{AB}+2C_6\beta_{AB}
+[0.5C_3-2 -(m+1)C_2]\verif
\geq0
\end{align*}
with $C_3\geq 4+2(m+1)C_2$. 

\item
$\disagree'_{AB}  = 0$ and $\verif > 2 \disagree_{AB}=0$.

Then, it is easy to see that all parties should have the same true meeting point
$\MPone$ in the past $\verif$ epochs.
Moreover, the $\gamma$ does not change.
We further consider two cases.
\begin{enumerate}
\item $\verif\geq2$ or some party transitions to $\MPtwo$.

If $\verif\geq2$, then all parties failed to transition to their $\MPone$ in the first half of the past $\verif$ epochs because of more than $0.2\cdot\frac{\verif}{2}$ bad votes.
If some party transitions to $\MPtwo$,
then that party should suffer at least $0.4\verif$ bad votes to $\MPtwo$.

Therefore, we have $\forall i\in AB: \beta_{i} > 0.1 \verif$, i.e. $\beta_{AB} > 0.1(m+1) \verif$.
From \eqref{eq:potential_eq} the potential after the transition stage is
$$
\Phi'=\agree' - C_7\gamma.
$$
By Lemma \ref{lem:changeOfAgree}, the decrease of $\agree$ is at most $2\verif-1$.
Then, we have
\begin{align*}
\Delta \Phi =&\Delta\agree-C_2\verif_{AB}+C_5\error_{AB}
+C_6\alpha_{AB}+2C_6\beta_{AB}
\\
\geq&-(2\verif-1)
-(m+1)C_2\verif+C_5\error_{AB}
+C_6\alpha_{AB}+0.2(m+1)C_6\verif
\\
=& 1+C_5\error_{AB}
+C_6\alpha_{AB}
+[0.2(m+1)C_6-2 -(m+1)C_2]\verif
\geq0
\end{align*}
with $C_6\geq 10/(m+1)+5C_2$.

\item $\verif=1$ and all parties transition to $\MPone$.

In this case,
we need to include the increase of the potential during the preceding
computation and verification stages.
Since no party does error transition, we should have $\forall i\in AB: \error_{i}=0$.
Moreover, all parties should do one dummy epoch simulation in the preceding computation stage.
Considering $\verif=1$, $\error=0$ and $\disagree_{AB}=0$, this means this epoch must suffer some corruption,
otherwise all parties should do one actual epoch simulation.
Thus, we only need to make sure the potential does not decrease too much in this epoch.
Since all parties do dummy simulation and transition to $\MPone$(=$\Point$) when $\verif=1$, $\agree_{AB}$, $\disagree_{AB}$ and $\gamma$ do not change.
According to \eqref{eq:potential_eq}, the only variable left that may decrease the potential is $\verif_{AB}$, since $\verif_{AB}$, $\error_{AB}$, $\alpha_{AB}$ and $\beta_{AB}$ are reset to zero after the transition.
But it is easy to see that all parties have $\verif_{AB}=0$ before the epoch.
Therefore, the potential does not decrease.

\end{enumerate}

\item
\label{case:gammeDecrease}
$\disagree'_{AB}  = 0$ and $\verif > 2 \disagree_{AB}>0$.

Then, by Lemma \ref{lem:haveATrueMP},
when $\forall i\in AB: \verif_i=\frac{\verif}{2} > \disagree_{AB}$,
all parties should have a true meeting point that is a multiple of $\frac{\verif}{2}$ in the past.

We further consider two cases.
\begin{enumerate}

\item Every party has saved that true meeting point.

A failed transition to it means that
there are more than $0.1 \cdot \frac{\verif}{2}$ bad votes.
Therefore, we have $\forall i\in AB: \beta_{i} > 0.05 \verif$, i.e. $\beta_{AB} > 0.05(m+1) \verif \geq 0.1\verif$.
Thus, the $\gamma$ does not change.

From \eqref{eq:potential_eq} the potential after the transition stage is
$$
\Phi'=\agree' - C_7\gamma.
$$
By Lemma \ref{lem:changeOfAgree}, the decrease of $\agree$ is at most $2\verif-1$.
Then, we have
\begin{align*}
\Delta \Phi =&\Delta\agree+C_3\disagree_{AB}-C_2\verif_{AB}+C_5\error_{AB}
+C_6\alpha_{AB}+2C_6\beta_{AB}
\\
\geq&-(2\verif-1)
+C_3\disagree_{AB} -(m+1)C_2\verif+C_5\error_{AB}
+C_6\alpha_{AB}+0.1(m+1)C_6\verif
\\
=& 1+C_3\disagree_{AB}+C_5\error_{AB}
+C_6\alpha_{AB}
+[0.1(m+1)C_6-2 -(m+1)C_2]\verif
\geq0
\end{align*}
with $C_6\geq 20/(m+1)+10C_2$.

\item  Some party has not saved that true meeting point.

From \eqref{eq:potential_eq} the potential after the transition stage is
$$
\Phi'=\agree' - C_7\gamma',
$$
and thus
\begin{align*}
\Delta \Phi =&\Delta\agree+C_3\disagree_{AB}-C_2\verif_{AB}+C_5\error_{AB}
+C_6\alpha_{AB}+2C_6\beta_{AB}-C_7\Delta\gamma.
\end{align*}
By Lemma \ref{lem:changeOfAgree}, the decrease of $\agree$ is at most $2\verif-1$.
Moreover, by Lemma \ref{lem:valuesForMissingMP2},
we have $\alpha_{AB} \geq 0.1\verif$ or $\gamma\geq0.25\verif$ in the past when $\forall i\in AB: \verif_i=\frac{\verif}{2} > \disagree_{AB}$.

If $\alpha_{AB} +\beta_{AB}  < 0.1\verif$,
then $\Delta\gamma=-0.25\verif$ according to how $\gamma$ decreases.
By Lemma \ref{lem:valuesForMissingMP2} we should have large enough $\gamma\geq0.25\verif$.
Then, we have
\begin{align*}
\Delta \Phi 
\geq&-(2\verif-1)
+C_3\disagree_{AB} -(m+1)C_2\verif+C_5\error_{AB}
+C_6\alpha_{AB}+2C_6\beta_{AB}+0.25C_7\verif
\\
=& 1+C_3\disagree_{AB}+C_5\error_{AB}+C_6\alpha_{AB}
+2C_6\beta_{AB}
+[0.25C_7-2 -(m+1)C_2]\verif
\geq0
\end{align*}
with $C_7\geq 8+4(m+1)C_2$.

Otherwise, if $\alpha_{AB} +\beta_{AB} \geq 0.1\verif$,
then $\Delta\gamma=0$
and we have
\begin{align*}
\Delta \Phi 
\geq&-(2\verif-1)
+C_3\disagree_{AB} -(m+1)C_2\verif+C_5\error_{AB}
+0.1C_6\verif + C_6\beta_{AB}
\\
=& 1+C_3\disagree_{AB}+C_5\error_{AB}
+C_6\beta_{AB}
+[0.1C_6-2 -(m+1)C_2]\verif
\geq0
\end{align*}
with $C_6\geq 20+10(m+1)C_2$.

\end{enumerate}

\end{enumerate}

\end{enumerate}

Combining the proof of Lemma \ref{lem:potentialIncreaseVCPhase},
if there exists at least one corruption or short hash collision,
the potential decreases by at most  $(m+1)(2C_6-C_2)$,
which happens when $\forall i,j\in AB: \verif_i = \verif_j$.
Moreover, if there is no corruption or short hash collision,
the potential increases as required in the statement.
\end{proof}

Actually, Case \ref{case:gammeDecrease} of Lemma~\ref{lem:potentialIncreaseEpoch} contains
the argument that $\gamma$ is always non-negative.
However, for completeness, we extract that part out in the following lemma.

\begin{lemma}
\label{lem:gammaDecrease}
The variable $\gamma$ is always non-negative.
\end{lemma}
\begin{proof}
Initially, $\gamma = 0$.  Hence, it suffices to check that whenever $\gamma$ is decreased,
it never drops below 0.  Recall that the condition for $\gamma$ to decrease
is as follows:
\emph{
All parties do meeting point transitions with $\forall i\in AB: 0<2\disagree_{AB}<\verif_{i}=\verif$ and $\alpha_{AB}+\beta_{AB} < 0.1\verif$ before the transition and with $\disagree_{AB}=0$ after the transition.  Then, $\gamma$ decreases by $0.25\verif$.
}

Since $\forall i\in AB: 2\disagree_{AB}<\verif_{i}=\verif$, by Lemma \ref{lem:haveATrueMP},
when $\forall i\in AB: \verif_i=\frac{\verif}{2} > \disagree_{AB}$,
all parties should have a true meeting point that is a multiple of $\frac{\verif}{2}$ in the past.
We consider the following two cases.

\begin{enumerate}
\item
Every party has saved that true meeting point.

A failed transition to it means that
there are more than $0.1 \cdot \frac{\verif}{2}$ bad votes.
Therefore, we have $\forall i\in AB: \beta_{i} > 0.05 \verif$, i.e. $\beta_{AB} > 0.05(m+1) \verif \geq 0.1\verif$.
Thus, the $\gamma$ does not change in this case.

\item 
Some party has not saved that true meeting point.

By Lemma \ref{lem:valuesForMissingMP2},
we have $\alpha_{AB} \geq 0.1\verif$ or $\gamma\geq0.25\verif$ in the past when $\forall i\in AB: \verif_i=\frac{\verif}{2} > \disagree_{AB}$.
Since $\alpha_{AB}+\beta_{AB} < 0.1\verif$, we should have large enough $\gamma\geq0.25\verif$ when $\gamma$ decreases.
\end{enumerate}
\end{proof}

\subsection{Upper Bound on Potential Function}

We next prove the third statement of Lemma~\ref{lemma:potential_tech}, i.e.,
the potential $\Phi$ cannot grow too fast. In particular, it grows naturally by one per epoch when a correct computation step is performed. On the other hand, any corruption also cannot increase this too much.

\begin{lemma}
[Upper Bound on Potential Function:
Similar to Lemma 7.5 in~\cite{Haeupler14}]
\label{lem:finalPotential}
The final potential $\Phi$ after $R$ epochs satisfies $\Phi \leq \agree \leq R$. 
\end{lemma}
\begin{proof}
Since $\agree $ can increase at most by one in each epoch,
we have $\agree  \leq R$. It suffices to consider the following cases.

\noindent \textbf{Case 1:}
When $\exists i,j\in AB: \verif_i \neq \verif_j$,
from \eqref{eq:potential_neq} we have 
$$\Phi=\agree - C_3 \disagree_{AB} - 0.9 C_4  \verif_{AB}  +  C_4  \error_{AB} -  C_6  \alpha_{AB}- C_6\beta_{AB}- C_7\gamma.$$
It is easy to see that at the end of each epoch, we have
$2\error_i<\verif_i$ for each party $i\in AB$,
i.e. $2\error_{AB}<\verif_{AB}$.
Thus, we have
$$\Phi<\agree - C_3 \disagree_{AB} - 0.9 C_4  \verif_{AB}  +  0.5 C_4  \verif_{AB} -  C_6  \alpha_{AB}- C_6\beta_{AB}- C_7\gamma
\leq\agree\leq R.$$

\noindent \textbf{Case 2:}
When $\forall i,j\in AB: \verif_i = \verif_j= \verif$,
from \eqref{eq:potential_eq} we have 
$$\Phi=\agree - C_3  \disagree_{AB} + C_2  \verif_{AB}  - C_5  \error_{AB} - C_6 \alpha_{AB}-2C_6\beta_{AB}- C_7\gamma.$$
Notice that if $\verif\leq 2\disagree_{AB}$,
then $- C_3  \disagree_{AB} + C_2  \verif_{AB} \leq-(m+3)\verif+(m+1) \verif\leq0$, i.e. $\Phi\leq\agree\leq R$.
Thus, we should consider $\verif > 2\disagree_{AB}$ next.
Let $j$ be an integer such that $2^j \leq \verif < 2^{j+1}$
and $t$ be the epoch in the past when $\verif$ grows to be $2^j=\verif^*$.
Then, the increase of the potential from epoch $t$ to the last epoch is at most $C_2(m+1)(\verif - \verif^*)< (m+1)\verif^*$.
Since $2^{j+1}> \verif > 2\disagree_{AB}$,
we have $2^j=\verif^*> \disagree_{AB}$.
Then, by Lemma \ref{lem:haveATrueMP},
there is a true meeting point $p^*$ that is a  multiple of $\verif^*$ during epoch $t$.

We considering the following cases:
\begin{enumerate}
\item  Every party has saved $p^*$.

A party fails to transition to the true meeting point $p^*$
because of more than 0.1$\verif^*$ corrupted votes,
i.e. $\forall i\in AB: \beta_i > 0.1\verif^*$.
Thus, after adding the following increase $(m+1)\verif^*$ of the potential,
the final potential is 
\begin{align*}
\Phi&<\agree - C_3  \disagree_{AB} + (m+1)\verif^*  - C_5  \error_{AB} - C_6 \alpha_{AB}-0.2C_6(m+1)\verif^*- C_7\gamma+(m+1)\verif^*\\
&=\agree - C_3  \disagree_{AB}  - C_5  \error_{AB} - C_6 \alpha_{AB}- C_7\gamma -(0.2C_6-2)(m+1)\verif^*\\
&<\agree\leq R.
\end{align*}

\item  Some party has not saved $p^*$.

This is because $p^*$ is the missing $\MPtwo$ for that party.
Then, by Lemma \ref{lem:valuesForMissingMP2}
we should have $\alpha_{AB} \geq 0.2\verif^*$ or $\gamma\geq0.5\verif^*$.
Thus, after adding the following increase  $(m+1)\verif^*$ of the potential,
the final potential is 
\begin{align*}
\Phi&\leq\agree - C_3  \disagree_{AB} + (m+1)\verif^*  - C_5  \error_{AB} - 2C_6 \beta_{AB}-0.2C_6\verif^*- C_7\gamma+(m+1)\verif^*\\
&=\agree - C_3  \disagree_{AB}  - C_5  \error_{AB} - C_6 \beta_{AB}- C_7\gamma -(0.2C_6-2m-2)\verif^*\\
&<\agree\leq R.
\end{align*}
or
\begin{align*}
\Phi&\leq\agree - C_3  \disagree_{AB} + (m+1)\verif^*  - C_5  \error_{AB} - C_6 \alpha_{AB} - 2C_6 \beta_{AB}- 0.5(4m+12)\verif^*+(m+1)\verif^*\\
&=\agree - C_3  \disagree_{AB}  - C_5  \error_{AB} - C_6 \alpha_{AB} - 2C_6 \beta_{AB} - 4\verif^*\\
&<\agree\leq R.
\end{align*}
\end{enumerate}

\ignore{
Furthermore, at the end of an iteration it always holds that $2 \error < \verif$  which implies that $-0.9 C_4 \verif_{AB} + C_4 \error_{AB} \leq - 0.4 \verif_{AB} \leq 0$ which leads to the potential $\Phi$ being at most $R$ if .
The only way to have a potential larger than $R$ is therefore if $\forall i,j\in AB: \verif_i = \verif_j$ and $\verif_{AB}$ is very large while
$\agree_{AB},\error_{AB},\alpha_{AB},\beta_{AB},\gamma$ are 0.
However, this is not possible without a large number of errors.
The only way to achieve this is that 

More precisely, in order to have $\Phi \leq R + x$, it has to be true that $x \leq C_2 \verif_{AB} $.
However, the only way for $\verif_{AB}$ to be larger than $2 \disagree_{AB} $ is if at least ten percent of the votes in the last $\verif_{AB} - 2 \disagree_{AB} $ rounds were corrupted to appear non-matching.
These kind of corruptions can furthermore not be caused by a short hash collision and therefore must be due to an error which implies that $10\% \cdot (\verif_{AB} - 2 \disagree_{AB} ) < 2 mn \eps$ or $(\verif_{AB} - 2 \disagree_{AB} ) < 20 mn \eps$.
Putting this together gives $x \leq 20 C_2 mn \eps$ and therefore $\Phi \leq R + 20 C_2 mn \eps$ as desired.
}

Therefore, we conclude that the final potential $\Phi$ after $R$ epochs is at most $\agree \leq R$.
\end{proof}

\subsection{Bounding the Number of Short Hash Collisions}

Next, we prove the second statement of Lemma~\ref{lemma:potential_tech},
which bounds the number of epochs with short hash collisions.
Since the randomness for short hashes are shared for epochs in the same phase,
we use the following version of Chernoff Bound.

\ignore{

\begin{lemma}
[new]
\label{lem:longHashCollisionProb}
Given the failure probability $\delta$, the probability that no long hash collision happens among all $R$ epochs is at least $1-\delta$.
\end{lemma}

\begin{proof}
Applying a union bound over all parties, when $\rho$-biased randomness is used, 
the probability that any two epochs have long hash collision is at most
$(m+1)(2^{-o}+\rho)
=(m+1)(2^{-o}+2^{-o})
=2^{-\Theta(\log n\log \frac{1}{\delta})}$.
Let $A_{ij}$ be the event that the $i$th epoch has long hash collision with the $j$th epoch.
Then, we have $\textnormal{Pr}[A_{ij}] \leq 2^{-\Theta(\log n\log \frac{1}{\delta})}$.
Using a union bound, we have the probability that
the $i$th epoch has long hash collision with any of the preceding $i-1$ epochs is at most $\sum_{j=1}^{i-1}\textnormal{Pr}[A_{ij}]=(i-1)2^{-\Theta(\log n\log \frac{1}{\delta})}$.
Thus, taking a union bound over all epochs,
we have the probability that there exists at least one long hash collision among all $R$ epochs is at most
$\sum_{i=1}^{R}(i-1)2^{-\Theta(\log n\log \frac{1}{\delta})}
=\Theta(n^2\epsilon)2^{-\Theta(\log n\log \frac{1}{\delta})}
\leq2^{-\Theta(\log \frac{1}{\delta})}
\leq\delta$
as required.
\end{proof}
}

\begin{theorem}
[Chernoff Bound \cite{Chernoff1952}]
\label{the:chernoff}
Let $X_1,..,X_k$ be independent random variables such that
$X_i\in[0,z]$ for a positive constant $z$.
Define $X=\sum_iX_i$ and $\mu=\textnormal{E}[X]$.
Then, for any $\sigma\geq 1$,
$$\textnormal{Pr}[X\geq(1+\sigma)\mu]\leq\exp(-\frac{\sigma\mu}{3z}).$$
\end{theorem}

Now we can show that the number of short hash collisions in 
Algorithms \ref{alg:ComputeOblivious} and \ref{alg:AliComputeOblivious} is small as follows.

\begin{lemma}
[Number of Epochs with Short Hash Collision:
Similar to Lemma 7.6 in~\cite{Haeupler14}]
\label{lem:shortHashCollisionNum}
By choosing the length of the short hash to be
$c=\Theta(\log m)$,
the probability that there are more than 
$\Theta( m n \eps)$ epochs with short hash collisions is at most $\exp(-\Theta(n^{0.25} \eps/I))$.
\end{lemma}

\begin{proof}
As observed in~\cite{Haeupler14}, short hash collisions only matter during inconsistent epochs.
Let $d$ be the number of inconsistent epochs.
Hence, we only need to consider the case when $d \geq \Theta(m n \epsilon)$.
Let $h$ be the number of inconsistent epochs with short hash collisions.

In an inconsistent epoch,
each of the $m$ channels compares 4 short hashes (each having $c$ bits),
which are computed using $2^{-c}$-biased randomness.
Hence, the probability~$p$ that there exists a short hash collision in an inconsistent epoch
satisfies $p \leq O(\frac{m}{2^c})$.


Lemma \ref{lem:potentialIncreaseEpoch} shows that the potential increases at least by 1 in a consistent epoch and by 2 in an inconsistent epoch if no error or short hash collision happens.
Otherwise, the potential decreases at most by a fixed constant $C^{-}=(m+1)(2C_6-C_2) = \Theta(m^4)$.
Assuming $R r \leq 2n$ as in Lemma~\ref{lemma:correct_sim},
the number of epochs with corruptions is at most $2mn\eps$.

Hence, the total potential change during inconsistent epochs is therefore at least $2 (d - h - 2 mn \eps) - C^{-} (h + 2 mn \eps) $
while the potential accumulated in consistent epochs is at least $R - d - 2mn\eps - C^{-} (2mn\eps)$.
Hence, the final potential is at least
$R + d -2h - 6 mn \eps - C^{-} (h + 4 mn \eps)$.

From Lemma \ref{lem:finalPotential}, we get however that the total potential is at most $R $.
Together this implies $d \leq (2 + C^{-}) h + \Theta(mn \eps)$.
Thus, since $d \geq  \Theta(mn \eps)$ is sufficiently large,
it follows that  $\frac{h}{d} \geq \frac{1}{2(2 + C^{-})} \geq 2 p$,
where the last inequality holds if we set $c = \Theta(\log m)$ large enough
such that $p = \Theta(\frac{1}{m^4})$.

Finally, observe that in the $d$ inconsistent epochs, the $I$ 
epochs in the same phase share the same randomness for short hashes,
while epochs in different phases use independent randomness.
Therefore, by the version of Chernoff Bound in Theorem~\ref{the:chernoff},
$\Pr[h \geq 2 dp] \leq \exp(- \Theta(\frac{dp}{I})) \leq \exp(-\Theta(\frac{n^{0.25} \epsilon}{I}))$,
where the last inequality follows because we assume $m \leq n^{0.25}$.

\end{proof}

%% file: coding.tex
\section{Subroutines for Pre-shared Randomness}
\label{sec:random}

In this section, we describe the subroutines to handle pre-shared randomness used for hashing.
Most of the ideas are standard in the literature, but we need to use variants that have low memory usage.

\subsection{Robust Transmission}

\ignore{
\begin{lemma}[From \cite{lin2004}]
\label{lem:dmin}
An error correcting code with minimum distance $d_{min}$ is capable of correcting any pattern of $\floor{(d_{min}-1)/2}$ or fewer random errors, where  $\floor{(d_{min}-1)/2}$ is also called the random-error-correcting capability of the code.
\end{lemma}
}

As in~\cite{Haeupler14}, we use error correcting code.  In particular,
we use the following variant.

\begin{lemma}[Justesen Code~\cite{Justesen1972,MacWilliams2004}]
\label{lem:Justesencode}
Given any positive integers $\ell$ and $t$, there exists a Justesen Code encoding a message of $\ell$ bits to a codeword of $\Theta(\ell+t)$ bits such that the codewords have minimum distance at least $2t + 1$.
In particular, the original message can be retrieved with at most $t$ corrupted bits.
\end{lemma}

Recall that our goal is to transmit an $\ell$-bit message over a channel that corrupts at most $t$ bits.
Justesen code in Lemma~\ref{lem:Justesencode} can directly achieve this,
but it takes $\Omega(\ell + t)$ bits of memory,
which might be undesirable when $t \gg \ell$.  The subroutine in Algorithm~\ref{alg:RobustTransmit}
uses $O(\ell + \log \frac{t}{\ell})$ bits of memory, where Alice is the sender and Bob is the receiver.

For the case $t \geq \ell$, the protocol uses a Justesen code with minimum distance at least
$2 \ell + 1$, and the sender repeatedly transmits the encoding of the original message 
for $2\ceil{\frac{t}{\ell}}+1$ iterations.  If the channel corrupts at most $t$ bits,
then the receiver fails to retrieve the original message for at most $\floor{\frac{t}{\ell}}$ iterations.
Therefore, the Boyer–Moore majority vote algorithm \cite{Boyer1991} in Line \ref{line:beginMJRTY} to \ref{line:endMJRTY} will be able to correctly identity the original message.

\ignore{
%
Using a Justesen code with minimum distance at least $2l+1$,
by Lemma \ref{lem:dmin} the adversary has to corrupt more than $l$ bits to corrupt a codeword in Line \ref{line:corruptedcodeword}.
Otherwise, Bob can recover the original randomness from Alice by decoding the codeword.
Since the tolerated number of corruptions is $t$ bits,
Bob may not recover the randomness of at most $t/l$ iterations.
Thus, Bob can recover the randomness in more than half of the $2\ceil{t/l}+1$ iterations, which guarantees his $\core^*$ is the same as Alice's at the end of the algorithm. 
Note that we adopt the Boyer–Moore majority vote algorithm \cite{Boyer1991} in Line \ref{line:beginMJRTY} to \ref{line:endMJRTY}.
}

\begin{algorithm}[htb!]
\caption{\textsc{RobustSend}($\core,t$) and \textsc{RobustReceive}($\ell, t$)}
\label{alg:RobustTransmit}
\begin{small}

\KwIn{Alice wishes to transmit an $\ell$-bit message $\core$ to Bob
over a channel that can corrupt at most $t$ bits.}

If $t < \ell$, then use the Justesen code in Lemma~\ref{lem:Justesencode} directly;
otherwise, do the following.


We use the Justesen code that encodes an $\ell$-bit message into
an $\Theta(\ell)$-bit codeword such that the minimum distance is at least $2 \ell + 1$;
let \textsc{Encode} and \textsc{Decode} be the functions associated with the encoding.

%


\SetKwProg{myproc}{procedure}{}{}
  
\myproc(\Comment*[f]{\normalfont Alice is the sender}){\textsc{RobustSend}$(\core,t)$}{ 
  
	Codeword $\codeword \gets$ \textsc{Encode}$(\core)$
	
	\For{$2\ceil{t/l}+1$ iterations}{
	
		Send $\codeword$ to Bob  \Comment*[f]{\normalfont Send codeword to Bob}
	}

	\nl \KwRet\;
}

\myproc(\Comment*[f]{\normalfont Bob is the receiver}){\textsc{RobustReceive}$(\ell,t)$}{ 

	$\core^* \gets \{0\}^{l}$; Counter $\counter \gets 0$ \label{line:beginMJRTY}

	\For(\Comment*[f]{\normalfont Less than half of the codewords can be corrupted}){$2\ceil{t/l}+1$ iterations}{
	
		Receive $\codeword'$ from Alice	\Comment*[f]{\normalfont Receive (corrupted) codeword from Alice} \label{line:corruptedcodeword}

		$\core' \gets$ \textsc{Decode}$(\codeword')$

		\If{$\counter=0$}{
			$\core^* \gets \core'$; $\counter \gets 1$
		}
		\ElseIf{$\core'=\core^*$}{
	
			$\counter \gets \counter +1$
		}
		\Else{
			$\counter \gets \counter -1$ \label{line:endMJRTY}
		}

	}
\Return{\emph{Received message} $\core^*$ }

}

\end{small}
\end{algorithm}

\begin{lemma}
\label{lem:numOfBitsRobust}
The total number of bits sent in Algorithm \ref{alg:RobustTransmit} is $\Theta(\ell+t)$.
The memory usage is $O(\ell + \log \frac{t}{\ell})$.
\end{lemma}

\begin{proof}
For $t < \ell$, the statement follows immediately. Otherwise,
only one codeword of $\Theta(l)$ bits is exchanged during each of the $2\ceil{t/l}+1$ iterations.
Thus, the total number of bits exchanged is $\Theta(\ell + t)$.

For the memory usage, each codeword takes $O(\ell)$ bits and each counter takes $O(\log \frac{t}{\ell})$ bits.
\end{proof}

\subsection{Generating Biased Randomness with Low Memory Usage}

As in~\cite{Haeupler14},
we use the technique of generating $\rho$-biased $q$-bit randomness
with a short seed~\cite{Naor1993}.

\noindent \textbf{High Level Intuition.}  The idea is to consider
a constant degree expander graph with $\Theta(q)$ vertices, each of which
is labeled with some suitably chosen $q$-bit vector.
To generate $\rho$-biased randomness, one first picks a starting vertex
uniformly at random (which takes $O(\log q)$ bits)
and performs a random walk for $l = O(\log \frac{1}{\rho})$ steps (which
takes $O(l)$ random bits).  Then, each visited vertex is selected 
independently with probability $\frac{1}{2}$ (which takes $l$ random bits).
Finally, a $q$-bit randomness is formed by summing up the vectors
corresponding to the chosen vertices.  Hence, given the expander graph,
the length of the seed is $O(l + \log q)$.

\noindent \textbf{Memory Usage.}  To use the above approach,
each party needs to have access to the expander graph, which takes $\Omega(q^2)$ bits
to store.  To ensure low memory usage, we do not rely on the expander graph.
In some sense, instead of using a random walk with a random starting vertex,
we just pick an independent vertex uniformly at random each time,
where the vector associated with a vertex does not need to be stored explicitly.
Therefore, in our approach, the length of the seed is $O(l \log q)$, which is slightly longer.

\noindent \textbf{Construction.}  For completeness, we outline the construction.
First, we assume that $q$ is a power of 2, which can be achieved by rounding to
the next power of 2.  We use the finite field $\F_q$ with order $q$,
where each element can be represented by $\log_2 q$ bits.
Hence, we also fix some bijection $\rank: \F_q \rightarrow [1..q]$
and label the elements $F_q = \{v_1, \ldots, v_q\}$ such that
$\rank(v_j) = j$.

\noindent \textbf{Polynomial Evaluation.}
Given a vector $\mathbf{b} \in \F_q^k$,
we index its coordinates with $[0..k-1]$.
For $v \in \F_q$, we denote $\EvalPoly(\mathbf{b}, v) := \sum_{j=0}^{k-1} \mathbf{b}_j v^j$.
The following result follows from standard linear algebra.

\begin{fact}[$k$-wise Independence]
Suppose $\mathbf{b} \in \F_q^k$ is selected uniformly at random.
Then, the induced sequence $(\EvalPoly(\mathbf{b}, v): v \in \F_q)$ of $q$ elements
is $k$-wise independent.
\end{fact}

Algorithm~\ref{alg:biasedRandom} is adapted from~\cite[Section 3.1.1]{Naor1993},
and we include it here for completeness.

\begin{algorithm}[!ht]
\caption{$\rho$-Biased $q$-Bit Randomness Generator}
\label{alg:biasedRandom}
\begin{small}
\KwIn{Finite field $\F_q = \{v_1, \ldots, v_q\}$
and bijection $\rank: \F_q \rightarrow [1..q]$ such that $\rank(v_j) = j$.}



\SetKwProg{myproc}{procedure}{}{}

\myproc{\textsc{RandInit}$(\core^*, q \in \Z, \rho \in (0,1))$}{

$l \gets \Theta(\log\frac{1}{\rho})$

Use the $\Theta(l \log q)$ bits in $\core^*$ to
initialize the following variables (independently uniformly at random)
from their respective domains.

Pick $\mathbf{a} \in \{0,1\}^l$, using $l$ bits.

\For{each $1\leq t \leq l$}
{

		Pick $\mathbf{b}^{(t)} \in \F^7_q$, using $7 \log q$ bits.

		Pick $\mathbf{c}^{(t)} \in \{0,1\}^{\log q}$, using $\log q$ bits.

		Pick $\mathbf{d}^{(t)} \in \F^2_q$, using $2 \log q$ bits.

		\ignore{
		Pick a subset $S^{(t)}\subseteq \{1,...,\log q\}$, using
    $\log q$ bits.

    \For {each $1\leq j \leq \log q$}
    {
        Compute $\mathbf{d}_j^{(t)}= |S^{(t)}\bigcap \{1,...,j\}| \mod 2 \in \{0,1\}$.
    }
		}
}

\Return{$\mathcal{S} \gets (\mathbf{a}, \mathbf{b}, \mathbf{c}, \mathbf{d})$}
}

\myproc{$\mathcal{S}$.\textsc{ExtractBit}$(i \in [1..q])$}{

			Set $v \in \F_q$ such that $\rank(v) = i$.

			\For{each $1\leq t \leq l$}{
        
						$u_t\gets \EvalPoly(\mathbf{b}^{(t)}, v) \in \F_q$


            \If{$\rank(u_t) \mod 2=0$\label{algline:umod2}}
            {
                    $r_t \gets 0$
            }
            \Else
            {

							$z_t \gets \EvalPoly(\mathbf{d}^{(t)}, v) \in \F_q$

							$j_t \gets \max\{1, \ceil{\log_2 \rank(z_t)}\}$

							$r_t \gets \mathbf{c}_{j_t}^{(t)} \in \{0,1\}$

            }
        }

\Return{$\sum_{t=1}^{l} \mathbf{a}_{t}r_t  \mod 2$}

}

\end{small}
\end{algorithm}

\begin{fact}
In Algorithm~\ref{alg:biasedRandom},
\textsc{RandInit} takes and stores  $\Theta(\log \frac{1}{\rho} \log q)$ bits,
and \textsc{ExtractBit} needs only 
$O( \log \frac{1}{\rho} \log q)$ bits of memory usage.
\end{fact}